\documentclass[letter]{article}
\usepackage[utf8]{inputenc}
\usepackage[OT4]{fontenc}
\usepackage[english]{babel}
\usepackage{graphicx}
\usepackage{subfigure}
\usepackage{fullpage}
\usepackage{comment}
\usepackage{xspace}
\usepackage{algorithmicx}
\usepackage{algorithm}
\usepackage[noend]{algpseudocode}
\usepackage{amsfonts,amsmath,amssymb,wasysym,stmaryrd}
\usepackage{url}
\usepackage{tikz}
\usepackage{caption}
\usepackage{enumitem}
\usepackage{color}
\usepackage{colortbl}
\usepackage{amsthm}
\usepackage{mathdots}

\usetikzlibrary{arrows}

 \makeatletter
 \let\@copyrightspace\relax
 \makeatother

\begin{document}

\newtheorem{theorem}{Theorem}[section]
\newtheorem{definition}[theorem]{Definition}
\newtheorem{lemma}[theorem]{Lemma}
\newtheorem{corollary}[theorem]{Corollary}
\newcommand{\pl}{\alpha}
\newcommand{\degth}{\delta}
\newcommand{\prob}[1]{P(#1)}
\newcommand{\dir}[1]{\vec{#1}}
\newcommand{\vdeg}[1]{deg(#1)}
\newcommand{\poly}{\mathrm{poly}}
\newcommand{\polylog}{\mathrm{polylog}}
\newcommand{\E}{\mathrm{E}\,}
\newcommand{\Var}{\mathrm{Var}\,}
\newcommand{\dpl}{PLB\xspace}
\newcommand{\outdeg}{\textrm{outdeg}}
\newcommand{\param}{\eta}
\newcommand{\ctr}{\textsc{CountTriangles}\xspace}
\newcommand{\ctrfmm}{\textsc{CountTrianglesFMM}\xspace}
\newcommand{\fullversion}[1]{#1}
\newcommand{\exabstract}[1]{}

\definecolor{powerlaw}{HTML}{F2F2F2}

\begin{titlepage}
\def\thepage{}
\thispagestyle{empty}

\title{Algorithmic Complexity of Power Law Networks}

\author{
Paweł Brach\footnote{University of Warsaw, \texttt{pawel.brach@mimuw.edu.pl}, Supported by the ERC StG PAAl project no. 259515.}
\and
Marek Cygan\footnote{University of Warsaw, \texttt{cygan@mimuw.edu.pl}}
\and
Jakub Łącki\footnote{University of Warsaw, \texttt{j.lacki@mimuw.edu.pl}. Jakub Łącki is a recipient of the Google Europe Fellowship in Graph Algorithms, and this research is supported in part by this Google Fellowship.}
\and
Piotr Sankowski\footnote{University of Warsaw, \texttt{sank@mimuw.edu.pl}, Supported by the ERC StG PAAl project no. 259515.}
}


\maketitle

\begin{abstract}
It was experimentally observed that the majority of real-world networks are scale-free and follow power law degree distribution.
The aim of this paper is to study the algorithmic complexity of such ``typical'' networks. The contribution of this work is twofold.

First, we define a deterministic condition for checking whether a graph has a power law degree distribution and experimentally validate it on real-world networks.
This definition allows us to derive interesting properties of power law networks. We observe that for exponents of the degree distribution in the range $[1,2]$ such networks exhibit double power law phenomenon that was observed for several real-world networks. Our observation indicates that this phenomenon could be explained by just pure graph theoretical properties.

The second aim of our work is to give a novel theoretical explanation why many algorithms run faster on real-world data than what is predicted by algorithmic worst-case analysis.
We show how to exploit the power law degree distribution to design faster algorithms for a number of classical P-time problems including transitive closure, maximum matching, determinant, PageRank and matrix inverse.
Moreover, we deal with the problems of counting triangles and finding maximum clique.
Previously, it has been only shown that these problems can be solved very efficiently on power law graphs when these graphs are random, e.g., drawn at random from some distribution. However, it is unclear how to relate such a theoretical analysis to real-world graphs, which are fixed.
Instead of that, we show that the randomness assumption can be replaced with a simple condition on the degrees of adjacent vertices, which can be used to obtain similar results.
Again, we experimentally validate that many real-world graphs satisfy our property.
As a result, in some range of power law exponents, we are able to solve the maximum clique problem in polynomial time, although in general power law networks the problem is NP-complete.

In contrast to previously done average-case analyses, we believe that this is the first ``waterproof'' argument that explains why many real-world networks are easier. Moreover, an interesting aspect of this study is the existence of structure oblivious algorithms, i.e., algorithms that run faster on power law networks without explicit knowledge of this fact or explicit knowledge of the parameters of the degree distribution, e.g., algorithms for maximum clique or triangle counting.

\end{abstract}
\end{titlepage}


\section{Introduction}\label{sec:introduction}
One of the most interesting observations in our understanding of complex networks is that for most large networks the degree
distribution closely resembles a power law distribution~\cite{RevModPhys.74.47}, i.e., the number of nodes
of degree $d$ is proportional to $d^{-\pl}$ for some $\pl> 1$. Such networks are called scale-free and many
models explaining their emergence have been proposed -- the most important one being the preferential attachment model~\cite{ba}.
The aim of this work is to study the algorithmic complexity of such ``typical'' networks and its contribution is twofold.

First, we define a deterministic condition for checking whether a graph has a power law degree distribution and show that many real-world networks satisfy it.
Graphs satisfying this condition are called power law bounded networks (\dpl).\footnote{For formal definition see Definition~\ref{def:dpl}.}
This definition allows us to derive new interesting properties of power law networks.
We observe that for $\pl \in [1,2]$ \dpl graphs with no parallel edges (simple graphs) need to exhibit double power law phenomenon.
This means that the degree distribution of vertices with sufficiently high degrees is different and has higher exponent.
This faster decay in the distribution was observed for some existing simple graphs and usually was attributed to
some complex processes~\cite{refId0}. Our results indicate that this phenomenon may have
a basic explanation that uses only pure graph theoretical properties. Essentially, we show that when $\pl \in [1,2]$ there are not enough
low degree vertices that can be connected to high degree vertices, and so the number of high degree vertices
needs to be lower and cannot be proportional to $d^{-\pl}$. This observation implies that
for $\pl \in [1,2]$ simple \dpl graphs have only $O(n^{2/\pl})$ edges. This contrasts with the
expected number of edges in power law multigraphs which is $O(n^{3-\pl})$.

The second contribution of this paper is the attempt to reduce the dichotomy in current research in algorithms, where two rarely interacting directions are pursued.
On one hand, theoreticians work on optimizing the performance of algorithms in the worst-case model. This is an important line of research
that has given us some beautiful algorithms and solutions. There are many success stories: a number of practically efficient
algorithms have been developed only thanks to this rigorous worst-case model, e.g., Dijkstra shortest paths algorithm. On the other hand,
there are problems where the best solutions that are used in practice have nothing in common with the state-of-art
algorithms proposed by theoreticians. This is clearly visible in the case of the Steiner tree problem, as exemplified
by last year's DIMACS implementations challenge. As shown, e.g., in~\cite{DIMACS} the algorithm of Byrka {\it et al.}~\cite{byrka}
with the best known theoretical approximation ratio, cannot be used on instances of larger size, because it is too inefficient. Moreover,
even on instances of smaller size it delivers worse results than the best metaheuristic approach based on local search~\cite{wer}.
The number of examples where heuristic approaches outperform ``worst-case'' algorithms is enormous. Intuitively, this is due to the fact that when one prepares for the worst case then the typical case will be handled in suboptimal way. Standard ways of overcoming this shortcoming are to work with stochastic models or random graphs, or use
smoothed analysis. For example, in online stochastic models it is sometimes possible to obtain better bounds on expected cost
of the algorithm than what is implied by worst-case competitive ratio~\cite{garg,setcover}. On the other hand, there are cases
where smoothed analysis allows us to obtain polynomial running time in expectation instead of exponential one~\cite{smoothed}.

However, the answers given by these stochastic models are still far from being satisfactory. Consider the rumor spreading process in a
social network, e.g., Twitter. It was observed that rumors spread extremely fast in such networks.
The paper~\cite{Doerr:2011:SNS:1993636.1993640} tries to give the following explanation for this observation.
Social networks have properties similar to networks obtained from preferential attachment model~\cite{ba},
so one tries to argue that fast spread of rumors in such random networks explains the rapid spread of rumors in real-world networks.
This explanation has the following shortcomings. First, it has been observed that although many properties
of social networks are explained well by this model, there are some properties that are not captured by it. For
example a better model is to use affiliation networks~\cite{Lattanzi:2009}. Even if social networks were random we would newer
know that we have a precise model for them, as we might always miss some important property. Hence, this
argument is far from explaining the observations. Second, there exists just one instantiation of any social network and there
is no way we can see distribution of all random Twitter networks that is needed for this argument. Besides, as there is just one example
of a social network it might be the unlucky one for the stochastic model that lies outside the whp statement. Finally
and most importantly, social networks are not random at all! They represent real-world ties, e.g., friendships
which are far from begin random.

In this paper we introduce the concept of a \dpl network, which gives a novel ``waterproof'' worst-case approach that overcomes the aforementioned problems and explains why many real-world networks
are easier.
We prove that on \dpl networks many problems have lower complexity than what is implied by classical solutions. The problems that we are able to solve faster include basic P-time problems: transitive closure, perfect matching, PageRank and counting triangles. Additionally, we show that the NP-hard problem of finding maximum clique allows a subexponential time algorithm in PLB networks. An interesting aspect of this study is the existence of structure oblivious algorithms, i.e., algorithms that run faster on \dpl networks without explicit knowledge of this fact. These structure oblivious algorithms shed some light on why some existing heuristic approaches are so efficient in practice, e.g., sorting vertices by degrees is the first step in many heuristic approaches to maximum clique problem~\cite{Carraghan1990375}.

Explaining why many algorithms work faster on real-world instances than what is predicted by worst-case analysis is one of the grand questions in algorithm that did not receive a plausible answer so far.
A notable example is the SAT problem~\cite{sat}.
Our paper gives a possible answer to this grand challenge and calls for further research in this direction. On one hand, we shall
search for faster solutions to other problems.
On the other hand, we believe that real-world power law networks have more worst-case graph properties that can be exploited in the design and analysis of algorithms.
In particular, we have observed that in a number of power law graphs with $\pl > 2$, every vertex of degree $k$ has $o(k)$ neighbors of degree at least $k$ (we say that the graph has \dpl neighborhoods).\footnote{For formal statement see Definition~\ref{def:dpl_neigh}.}

We have experimentally confirmed that this property is present in a number of real-world networks.
This property can effectively replace the randomness assumption about the graph that has been introduced in previous works and we use it to obtain faster algorithms for counting triangles and maximum clique problem.
In particular it implies that for $\pl>3$ our maximum clique algorithm works in polynomial time.
This observation clearly contrasts with the proof that the clique problem is NP-hard on power law networks for any $\pl>1$~\cite{lin}, and implies that it should be
possible to efficiently find maximum cliques in numerous real-world networks, in which $\pl>3$.

\subsection{Our Results and Related Work}
We study the algorithmic complexity of power law networks in a worst-case model.
Our work is somewhat related to the area of average-case analysis of algorithms, which tries to explain why some algorithm are
fast on real-world data.
However, we do not use the randomness of the data.
Instead, we identify graph properties that can be exploited to give efficient algorithms.
We stress that we are only interested in properties that can be decided deterministically.
We also show the our model is general, by proving that one of the basic random power law network
model generates \dpl graphs with high probability.



\paragraph{Counting Triangles}
The problem of finding or counting triangles in a graph can be solved in $O(n^{\omega})$ time or in
$O(m^{\frac{2\omega}{\omega+1}})$ time using fast matrix multiplication~\cite{triangles}. There has been some
work that tried to show faster algorithms for counting triangles in power law graphs.
Latapy~\cite{Latapy2008458} has shown two $O(mn^{1/\pl})$ time algorithms, where $m$ is the number of edges in the graph.
Moreover, Berry et. al~\cite{berry2014} have shown that in random power law graphs, generated by erased configuration model, triangles can be counted in $O(n\Delta^{7-3\pl})$ time, where $\Delta$ is the maximum vertex degree in the graph.
Since the model assumes that $\Delta/\sqrt{m} \leq 1/2$, for $\pl \in (2, 7/3)$ this gives a $O(n^{9/2-3/2\pl})$ time algorithm ($\pl > 2$ implies $m = O(n)$) and a linear time algorithm for $\pl > 7/3$.
However, as the authors admit this algorithm requires the graph to be random and does not fully apply to real-world graphs.
In addition, the assumption that $\Delta / \sqrt{m} \leq 1/2$ may be unrealistic, as it is satisfied in only few of the real-world networks that we have analyzed (see Table~\ref{table:plb_adjustment}).

We show that a very basic and widely used triangle counting algorithm works faster than what has been demonstrated by Latapy.
This simple algorithm processes nodes in increasing order of their degrees, computes the number of triangles incident to each vertex, and then removes the processed vertex.
A simple analysis shows that this algorithm runs in $O(n^{3/\pl})$ time for $1 < \pl < 3$, $O(n \log n )$ time
for $\pl = 3$, and $O(n)$ time for $\pl > 3$.
Additionally, for graphs with \dpl neighborhoods this algorithm runs in $O(n^{9/2-3/2\pl})$ time for $2 < \pl < 7/3$, and $\tilde{O}(n)$ time for $\pl \geq 7/3$.
These bounds visibly improve the running time of Latapy's algorithm for $\pl>2$ and match the results of Berry et. al~\cite{berry2014} (up to logarithmic factors) that have been obtained under full-randomness assumption. Moreover, when applied to random networks as in~\cite{berry2014}, our framework implies stronger whp bounds instead of bounds in expectations. We note that our algorithms are structure oblivious and do not need to know that the graph is \dpl or has \dpl neighborhoods to run in the above bounds. These running times are shown in Fig.~\ref{fig:times} and can be slightly improved by using fast matrix multiplication.

\paragraph{Maximum Clique}
The fastest algorithm for finding maximum cliques in general graphs runs in $O(1.2125^n)$ time~\cite{mis}.
Moreover, Chen et. al~\cite{Chen} have shown that maximum clique cannot be solved in subexponential time unless exponential
time hypothesis fails. We note that the maximum clique problem is NP-hard on power law graphs~\cite{lin}.
Janson, Łuczak and Norros~\cite{janson} have shown that for $\pl > 2$ maximal clique in a power law graph can be found in polynomial time and approximated for any $\pl$.
However, they assume that the graph is created using random Poissonian model.
In this paper we show that on \dpl graphs the problem can be solved in subexponential $\exp(O(n^{1/\pl}))$ time.
Additionally, when the graph has \dpl neighborhoods our algorithm runs in $\exp(O(n^{3/2 - \pl/2}\log n))$ time for $2 < \pl < 3$ and $O(\poly(n))$ time for $\pl > 3$.

\paragraph{Transitive Closure}
The transitive closure of a graph $G$ can be either computed in $O(nm)$ time by executing $n$ graph searches, or
in $O(n^{\omega})$ time using block recursion and fast matrix multiplication. We show that this
running time can be improved when $1<\pl<2$ -- see Fig.~\ref{fig:sank_times}.

\paragraph{Algebraic Matrix Algorithms}
There are two complexity results for the computation of the determinant of a $n\times n$ matrix $A$ over a finite fields:\footnote{We discuss here only the finite field case as it is the most relevant case for TCS.}
 (i) fast matrix multiplication to obtain $O(n^{\omega})$ time algorithm\footnote{$\tilde O(n^{\omega})$ is the time needed for a straight-line program to multiply two $n\times n$ matrices; $\omega$ is called matrix multiplication exponent. Currently $\omega< 2.373$~\cite{Williams}.} or (ii)
Wiedemann's approach that works in $\tilde{O}(nm)$ time, where $m$ is the number of nonzero entries in a matrix. 
We note that there are many heuristic approaches that are used in practice to speed up matrix computations, e.g., minimum 
degree algorithm~\cite{doi:10.1137/1031001}, but these ideas do not improve the worst-case complexities that are stated above. 
Here, we are only interested in obtaining a worst-case bound on the arithmetic complexity of these problems and therefore we will not review this rich body of 
literature. We note that our approach is related to minimum degree algorithm, because as the first step we partition 
the matrix into dense and sparse part according to the number of nonzero entries in each row or column. However, after this 
step novel algorithms are proposed that exploit the structure of the matrix. 

We will assume that
the non-zero structure of $A$ corresponds to an \dpl graph $G$, i.e., $a_{ij}\neq0$ if and only if $ij\in E(G)$. We are
able to show faster algorithms for the case when $1<\pl<2$. In particular our algorithm in the
case of symmetric matrices works in $O(n^{2+\frac{(\omega-2)(2-\pl)}{(\omega-2)\pl+3-\omega}})$ time -- see Fig.~\ref{fig:sank_times}
for the running time in the case of symmetric and general matrices.

Additionally, we show that with the same complexities
it is possible to solve linear system with matrix $A$, invert matrix $A$, and compute PageRank of a graph represented by $A$.
PageRank is a very simple version of the eigenproblem and our results could indicate that a general eigenproblem could be solved
faster on \dpl graphs. Developing such faster algorithms for eigenproblem, characteristic polynomial or even matrix rank is
left as an intriguing open problem.

\paragraph{Perfect Matching}
There are several algorithms known for finding perfect matching in general graphs: $O(\sqrt{n}m)$ time algorithm~\cite{mv}, $O(n^{\omega})$ time
algorithm~\cite{ms04} and $O(m^{10/7})$ time algorithm~\cite{Madry13}. Here, basing on our results for computing matrix determinant
we show an algorithm that improves over these results when $\pl<1.09$ -- see Fig.~\ref{fig:sank_times}. We conjecture, however, that an improvement is possible for $\pl \in [1,2]$.

\paragraph{Organization of the Paper}
The following part of this paper is organized as follows. \fullversion{In Section~\ref{sec:preliminaries} we introduce basic notation and show some general properties that we later use.}
In Section~\ref{sec:plb} we define the class of \dpl graphs and show their basic properties.
Then, in Section~\ref{sec:def_experiments} we verify our definitions on real-world data.
Section~\ref{sec:algorithms} analyses very simple algorithms for counting triangles and finding maximum clique on \dpl graphs.
Finally, in Section~\ref{sec:algebraic} we present more advanced algebraic algorithms for \dpl graphs that compute the transitive closure, find the perfect matching, and compute the determinant.

\fullversion{\section{Preliminaries}\label{sec:preliminaries}}
\exabstract{\paragraph{Preliminaries}}
Let $G$ be a graph.
Throughout the paper we use $n$ to denote the number of vertices in a graph, $d_k$ to denote the number of vertices of degree $k$, and $d_{\geq k}$ to denote the number of vertices of degree at least $k$.
It should be clear from the context, which graph we refer to.
We assume that the graphs we work with are simple, i.e., they do not contain multiple edges.
In the majority of problems that we study (e.g., transitive closure or maximal clique) multiple edges are not important and can be simply removed.
We assume that $\log n$ denotes the binary logarithm function.

\fullversion{
\begin{lemma}\label{lem:sum}
Let $1 \leq a \leq b$, for $a,b \in \mathbb{N}$, and let $c$ be a constant.
Then
\[
    \sum_{i=a}^b i^c = \begin{cases}
        O(b^{c+1}) & \text{if } c > -1\\
        O(\log(b/a)) & \text{if } c = -1\\
        O(a^{c+1}) & \text{if } c < -1\\
    \end{cases}
\]
\end{lemma}

Note that, throughout the paper we assume that for $b < a$, and any function $f$, $\sum_{i=a}^b f(i) = 0$.
\begin{proof}
For $i \geq a \geq 1$ we have $\lfloor i \rfloor^c = O(i^c)$.
Thus,
\[
    \sum_{i=a}^b i^c = \int_a^{b+1} \! \lfloor x \rfloor^c \, \mathrm{d}x = O(1) \int_a^{b+1} \! x^c \, \mathrm{d}x
\]

For $c \neq -1$ we have
\[
\int_a^{b+1} \! x^c \, \mathrm{d}x = \frac{1}{c+1}((b+1)^{c+1} - a^{c+1})
\]
If $c > -1$, then $\frac{1}{c+1} > 0$, so we we can bound the expression by $O((b+1)^{c+1}) = O(b^{c+1})$.
Otherwise, if $c < -1$, then $\frac{1}{c+1} < 0$, so we can bound it by $O(a^{c+1})$.
It remains to consider the case when $c = -1$:
\[
\int_a^{b+1} \! x^c \, \mathrm{d}x = \int_a^{b+1} \! x^{-1} \, \mathrm{d}x = \log(b+1) - \log a = O(\log(b/a))
\]
\end{proof}
}

\fullversion{
We also have a reverse relation:
\begin{lemma}\label{lem:sum-rev}
Let $1 \leq a \leq b/2$, for $a,b \in \mathbb{N}$, and let $c > 0$ be a constant.
Then $a^{-c} = O(\sum_{i=a}^b i^{-c-1})$.
\end{lemma}

\begin{proof}
\[
    \sum_{i=a}^b i^{-c-1} \geq \int_a^{b+1} \! x^{-c-1} \, \mathrm{d}x = \frac{1}{c} (a^{-c} - (b+1)^{-c}) \geq \frac{1}{c} (a^{-c} - (a/2+1)^{-c}) = \Omega(a^{-c})
\]

\end{proof}
}

\fullversion{
\begin{lemma}\label{lem:sum-t}
Let $c > 0$, $\pl > 1$, and $\degth \geq 1$.
Then
\[
    \sum_{i=1}^\degth (i+t)^{-\pl}i^c = \begin{cases}
        O(\degth^{c+1-\pl}) & \text{if } c > \pl - 1\\
        O(\log \degth) & \text{if } c = \pl -1\\
        O((t+1)^{c+1-\pl}) & \text{if } c < \pl -1\\
    \end{cases}
\]
\end{lemma}

\begin{proof}
If $c-\pl \geq -1$ we simply use the fact that $(i+t)^{-\pl} \leq i^{-\pl}$ and obtain
\[
    \sum_{i=1}^\degth (i+t)^{-\pl}i^c \leq  \sum_{i=1}^\degth i^{c-\pl}.
\]
By Lemma~\ref{lem:sum} this is equal to $O(\degth^{c+1-\pl})$ for $c > \pl -1$, and equal to $O(\log \degth)$ for $c = \pl -1$.

Now consider the case when $c < \pl -1$.
\begin{align*}
\sum_{i=1}^\degth (i+t)^{-\pl}i^c & \leq \sum_{i=1}^t (i+t)^{-\pl}i^c + \sum_{i=t+1}^\degth (i+t)^{-\pl}i^c\\
        & \leq O(1)\sum_{i=1}^t (2t+1)^{-\pl}i^c + \sum_{i=t+1}^\degth i^{c-\pl}\\
        & = O(1) (t+1)^{-\pl}\sum_{i=1}^t i^c + \sum_{i=t+1}^\degth i^{c-\pl}\\
        & = O((t+1)^{c+1-\pl}) + O((t+1)^{c+1-\pl})\\
        & = O((t+1)^{c+1-\pl}).
\end{align*}
\end{proof}
}

\begin{lemma}
\label{lem:inc_edges}
Let $G$ be a graph and $k \geq 0$. The number of edges of $G$ is at most $\sum_{i=1}^{n-1} d_{\geq i}$.
\end{lemma}
\fullversion{

\begin{proof}
Observe that a a vertex of degree $k$ is counted $k$ times in the sum.
Thus, the sum is equal to the total degree of all vertices, which is twice the number of edges.
\end{proof}
}

\section{Power law bounded networks}\label{sec:plb}
In this section we introduce our definition of a power law bounded network.
There are multiple definitions of power law networks.
Some of them state that in a power law network the number of vertices of degree $k$ is proportional to $k^{-\pl}$ for some parameter $\pl$~\cite{Aiello01}.
In other cases power law is defined with respect to random graphs and only talks about expected degrees of vertices~\cite{ba,RevModPhys.74.47}.
Both these approaches may not be applied to the analysis of algorithms running on real-world networks.
The first one suffers from two serious drawbacks.
First, it is often not stated in a formal way.
Second, it seems that it effectively disallows even a single vertex with high degree.
On the other hand the stochastic definition can only be applied to graphs randomly drawn from some distribution.
This is not the case for real-world graphs, which are fixed.

We introduce the concept of a power law bounded network, which captures the power law behavior of degree distribution that is necessary for the analysis of algorithms.
At the same time it is weak enough to cover many real-world graphs.
Note that this definition for $t=0$ is similar to the one in~\cite{berry2014}.
The main difference is that we do not impose any lower bounds on the numbers of vertices of given degrees.

\begin{definition}
\label{def:dpl}
Let $G$ be an undirected $n$-vertex graph and $c_1 > 0$ be a universal constant.
We say that $G$ is power law bounded (\dpl) for some parameters $1 < \pl = O(1)$ and $t \geq 0$ if for every integer $k \geq 0$, the number of vertices $v$, such that $\deg(v) \in [2^d, 2^{d+1})$ is at most
\[
    c_1 n (t+1)^{\pl-1} \sum_{i={2^d}}^{2^{d+1}-1} (i+t)^{-\pl}.
\]
In the following we say that $G$ is a \dpl{} graph with parameters $\pl$ and $t$.
\end{definition}

Note that the $(t+1)^{\pl-1}$ factor in the above definition is necessary to ensure that the sum of the above upper bounds over all $k$ is $O(n)$.
The above power law distribution that includes the shift by the parameter $t$ is called shifted power law~\cite{10.1371/journal.pone.0024926} and
was observed in different real-world networks. In particular, the parameter $t$ allows us to better fit the degree distributions in our experiments (see Section~\ref{sec:def_experiments}).
As our experiments show, in the networks that we have studied the value of $t$ is very small. However, in general it is unknown whether and how $t$ depends
on other parameters of the network and we are not aware of the models that would describe such dependence. A reasonable assumptions here seems to be
that $t = O(n^{\epsilon})$ for every $\epsilon>0$. However, when discussing some complexities of our algorithms we will for simplicity
sometimes assume that $t=O(\polylog n)$. Hence, the factors in the running time, that depend on $t$ are of secondary importance. In the introduction when discussing
our results we have assumed that $t=0$.

The exact set of graph that satisfy Definition~\ref{def:dpl} obviously depends on the choice of the constant $c_1$.
However, as we later show, many real-world graphs satisfy this definition for a small value of $c_1$, i.e., at most $5$.
At the same time, the running time dependency of our algorithms on $c_1$ is only polynomial.
The only exception is an algorithm for finding maximum clique, whose running time itself is super-polynomial.

\fullversion{Let us list some basic properties of \dpl{} graphs.}
\begin{lemma}
\label{lem:high_degree_num}
Let $G$ be a \dpl graph with parameters $\pl$ and $t$.
Then, $d_{\geq k} = O(n(t+1)^{\pl-1}(k+t)^{1-\pl}) = O(n(t+1)^{\pl-1}k^{1-\pl})$.
\end{lemma}

\fullversion{
\begin{proof}

Observe that $k' = 2^{\lfloor \log k \rfloor}$ is the smallest power of $2$ which is not greater than $k$, thus, $k'\leq k \leq 2k'$.
We bound the number of vertices, whose degree is at least $k'$, which is an upper bound on the number of vertices of degree at least $k$.

\begin{align*}
c_1 n (t+1)^{\pl-1} \sum_{i=k'}^{n-1} (i+t)^{-\pl} & \leq c_1 n (t+1)^{\pl-1} \sum_{i=k'+\lfloor t \rfloor}^{n-1+\lfloor t \rfloor} i^{-\pl}\\
       & = O(n (t+1)^{\pl-1} (k'+\lfloor t \rfloor)^{1-\pl})\\
       & = O(n (t+1)^{\pl-1} (k' + t)^{1-\pl})\\
       & = O(n (t+1)^{\pl-1} (k + t)^{1-\pl})
\end{align*}
\end{proof}
}

\fullversion{
The following lemma is used, e.g., to bound the running times of algorithms, which take $f(k)$ time to process a vertex of degree $k$, where $f$ is at most polynomial in its parameter.
Roughly speaking, it says that the running time of a polynomial algorithm running on a PLB network is asymptotically the same as the running time on a graph with an ideal power law distribution.

\begin{lemma}\label{lem:rtime_sum}
Let $G$ be a \dpl{} graph with parameters $\pl$ and $t$.
Let $d_i$ be the number of vertices of degree $i$ in $G$.
Let $f : \mathbb{N} \rightarrow \mathbb{N}$ be a nondecreasing function, such that for any $x, c \in \mathbb{N}$, $f(cx) \leq c^{O(1)} f(x)$.
Then, for every $k \geq 1$ we have $\sum_{i=1}^k d_i f(i) = O(1) n(t+1)^{\pl-1} \sum_{i=1}^{k} (i+t)^{-\pl} f(i)$.
\end{lemma}

\begin{proof}
Let us first derive an auxiliary inequality.
\begin{align*}
\sum_{i=2^j}^{i=2^j-1} d_i f(i) & \leq \sum_{i=2^j}^{i=2^{j+1}-1} d_i f(2^{j+1}) \\
       & = f(2^{j+1}) \sum_{i=2^j}^{i=2^{j+1}-1} d_i & \\
       & \leq O(1) n(t+1)^{\pl-1} \sum_{i=2^j}^{i=2^{j+1}-1} (i+t)^{-\pl} f(2^j)\\
       & \leq O(1) n(t+1)^{\pl-1} \sum_{i=2^j}^{i=2^{j+1}-1} (i+t)^{-\pl} f(i)
\end{align*}
Note that we introduce $O(1)$ to hide $c_1$ and the constant that comes from replacing $f(2^{j+1})$ by $f(2^j)$.
Let $k' = 2^{\lceil \log (k+1) \rceil} - 1$.
Thus $k \leq k' \leq 2k$ and $k' = 2^l-1$ for some integer $l$.

\begin{align*}
\sum_{i=1}^k d_i f(i) & \leq \sum_{i=1}^{k'} d_i f(i)\\
    & = \sum_{j=0}^{l-1} \sum_{i=2^j}^{2^{j+1}-1}  d_i f(i)\\
    & = \sum_{j=0}^{l-1} O(1) n(t+1)^{\pl-1} \sum_{i=2^j}^{i=2^{j+1}-1} (i+t)^{-\pl} f(i)\\
    & = O(1) n(t+1)^{\pl-1} \sum_{j=0}^{l-1}  \sum_{i=2^j}^{i=2^{j+1}-1} (i+t)^{-\pl} f(i)\\
    & = O(1) n(t+1)^{\pl-1} \sum_{i=1}^{k'} (i+t)^{-\pl} f(i)\\
    & \leq O(1) n(t+1)^{\pl-1} \sum_{i=1}^{2k} (i+t)^{-\pl} f(i)\\
    & = O(1) n(t+1)^{\pl-1} \sum_{i=1}^{k} ((2i-1+t)^{-\pl} f(2i-1) + (2i+t)^{-\pl} f(2i))\\
    & \leq O(1) n(t+1)^{\pl-1} \sum_{i=1}^{k} (i+t)^{-\pl} (f(2i-1) + f(2i))\\
    & \leq O(1) n(t+1)^{\pl-1} \sum_{i=1}^{k} (i+t)^{-\pl} f(i)\\
\end{align*}
\end{proof}
}

\fullversion{
By using Lemma~\ref{lem:rtime_sum} together with Lemma~\ref{lem:sum} we obtain the following
bound on the number of edges touching small degree vertices.
}

\begin{lemma}
\label{lem:small_degrees}
Let $G$ be a \dpl{} graph with parameters $\pl$ and $t$, where $\pl < 2$.
Then, the number of edges incident to at least one vertex of degree at most $k$ is $O(n(t+1)^{\pl-1} k^{2-\pl})$.
\end{lemma}

\fullversion{
\begin{proof}
The number of edges incident to at least one vertex of degree at most $k$ is $\sum_{i=1}^{k} d_i i$.
We use Lemma~\ref{lem:rtime_sum} with the identity function $f(i)=i$, obtaining
\begin{align}
\sum_{i=1}^{k} d_i i & = O(n(t+1)^{\pl-1} \sum_{i=1}^{k} (i+t)^{-\pl} \cdot i) \nonumber \\
  & = O(n(t+1)^{\pl-1} \sum_{i=1}^{k} i^{1-\pl}) \nonumber \\
  & = O(n(t+1)^{\pl-1} k^{2-\pl}). \label{eq:1}
\end{align}
In the first transformation we use the fact that $(i+t)^{-\pl} \le i^{-\pl}$,
whereas in the second one we use Lemma~\ref{lem:sum}.
\end{proof}
}

By combining Lemma~\ref{lem:high_degree_num} with Lemma~\ref{lem:inc_edges}, we obtain the following.

\begin{lemma}
\label{cor:basic_edges}
Let $G$ be a \dpl graph with parameters $\pl$ and $t$.
Then, the number of edges of $G$ of is (a) $O(n^{3-\pl}(t+1)^{\pl-1})$ for $1 < \pl <2$, (b) $O(n\log n(t+1))$ for $\pl=2$, (c) $O(n(t+1))$ for $\pl>2$.
\end{lemma}

\fullversion{
\begin{proof}
By Lemma~\ref{lem:inc_edges} the number of edges is at most $\sum_{i=k+1}^{n-1} d_{\geq i}$.
By Lemma~\ref{lem:high_degree_num}, $d_{\geq i} = O(n (t+1)^{\pl-1} (i+ t)^{-\pl+1})$.
We have
\begin{equation}
\label{eq:num_edges}
\sum_{i=1}^{n-1} O(n (t+1)^{\pl-1} (i + t)^{1-\pl}) = n (t+1)^{\pl-1} \sum_{i=1}^{n-1} O((i + t)^{1-\pl}).
\end{equation}
We now use Lemma~\ref{lem:sum}.
For $1 < \pl < 2$ we have $1 - \pl > -1$, so the sum can be bounded by $O((n+t)^{2-\pl}) = O(n^{2-\pl})$.
Putting it back into Equation~\ref{eq:num_edges}, we get $O(n^{3-\pl}(t+1)^{\pl-1})$.
For $\pl = 2$, we may bound the sum by $O(\log((n+t)/t)) = O(\log n)$, thus obtaining $O(n\log n(t+1))$.
Finally, for $\pl > 2$, we bound the sum by $O((t+1)^{2-\pl})$, so the number of edges is $O(n(t+1))$.
\end{proof}
}

What is interesting, for a \dpl graph with $1 < \pl < 2$, the bound on the number of edges given by Lemma~\ref{cor:basic_edges} is not tight.
In particular, the number of vertices with high degree (considerably greater than $n^{1/\pl}$) is polynomially smaller.
We say that a vertex is a high-degree vertex if its degree is more than $n^{1/\pl}$.
Each edge either connects two high-degree vertices or is incident to a low-degree vertex.
The number of edges of the first type is bounded, as there are few high-degree vertices, whereas the number of the edges of the second type is bounded by simply summing the degrees of low-degree vertices.
Note that this reasoning heavily depends on the fact that the graph is simple.
This is formalized in the following Lemma.

\begin{lemma}
\label{lem:alpha12}
Let $G$ be a \dpl graph with parameters $\pl$ and $t$, where $1 < \pl < 2$, and $k \geq n^{1/\pl}(t+1)^{1-1/\pl}$.
Then, $d_{\geq k} = O(n^{3-\pl}(t+1)^{(\pl-1)(3-\pl)} k^{\pl^2 - 3\pl + 1})$.
\end{lemma}

\fullversion{
\begin{proof}
We say that a vertex of degree at least $k$ is a {\it{high-degree}} vertex.
By Lemma~\ref{lem:high_degree_num}, $d_{\geq k} = O(n(t+1)^{\pl-1}k^{1-\pl})$.
We will use the fact that $G$ has no multiple edges to derive a stronger upper bound on $d_{\geq k}$.

We first bound the total degree of high-degree vertices, which we denote by $S$.
The edges, whose both endpoints have high degrees contribute at most $d_{\geq k}(d_{\geq k} - 1) \leq d_{\geq k}^2$ to $S$.
In addition, a low-degree vertex of degree $i$ contributes at most $\min(i, d_{\geq k})$.
Recall that by $d_i$ we denote the number of vertices of degree $i$.
Thus, we may bound $S$ by $d_{\geq k}^2 + \sum_{i=1}^{k-1} d_i \min(i, d_{\geq k})$.
We now apply Lemma~\ref{lem:rtime_sum}, using $f(i) = \min(i, d_{\geq k})$.

\begin{align*}
d_{\geq k}^2 + \sum_{i=1}^{k-1} d_i \min(i, d_{\geq k}) & \leq d_{\geq k}^2 + O(1)n(t+1)^{\pl-1}\sum_{i=1}^{k-1} i^{-\pl}  \min(i, d_{\geq k})\\
        & \le d_{\geq k}^2 + O(1)n(t+1)^{\pl-1}\left(\sum_{i=1}^{d_{\geq k}} i^{1-\pl} + \sum_{i=d_{\geq k}+1}^{k} i^{-\pl} d_{\geq k}\right)\\
        & = d_{\geq k}^2 + O(1)n(t+1)^{\pl-1}\left(O(d_{\geq k}^{2-\pl}) + O(d_{\geq k}^{2-\pl})\right)\\
        & = d_{\geq k}^2 + O(n(t+1)^{\pl-1}d_{\geq k}^{2-\pl})\\
\end{align*}
Note that when we split the sum into two sums, we use the assumed convention that for $a > b$, $\sum_{i=a}^b f(i) = 0$.
We now bound $d_{\geq k}$:
\[
d_{\geq k} = O(n(t+1)^{\pl-1}k^{1-\pl}) = O(n(t+1)^{\pl-1}n^{1/\pl-1}(t+1)^{(1-1/\pl)(1-\pl)}) = O(n^{1/\pl}(t+1)^{(\pl-1)/\pl}).
\]
This gives
\[
d_{\geq k}^2 = d_{\geq k}^{2-\pl} d_{\geq k}^\pl = d_{\geq k}^{2-\pl} O(n(t+1)^{\pl-1}).
\]
Hence, the total degree of high-degree vertices is
\begin{align*}
    d_{\geq k}^2 + O(n(t+1)^{\pl-1}d_{\geq k}^{2-\pl}) & = O(n(t+1)^{\pl-1}d_{\geq k}^{2-\pl})\\
   & = O(n(t+1)^{\pl-1}n^{2-\pl}(t+1)^{(\pl-1)(2-\pl)}k^{(1-\pl)(2-\pl)})\\
   & = O(n^{3-\pl}(t+1)^{(\pl-1)(3-\pl)} k^{(1-\pl)(2-\pl)}).
\end{align*}
To obtain the bound on the \emph{number} of high-degree vertices, we divide the obtained bound by $k$, which gives
$O(n^{3-\pl}(t+1)^{(\pl-1)(3-\pl)} k^{\pl^2 - 3\pl + 1})$.
\end{proof}
}

\begin{corollary}
\label{corollary-double-power-law}
Let $G$ be a \dpl graph with parameters $\pl$ and $t$, where $1 < \pl < 2$, and $k \geq n^{1/\pl}(t+1)^{1-1/\pl}$.
Moreover, assume that $1 < \pl < 2$.
Then, the number of vertices of degree between $k$ and $2k$ is $O(n^{3-\pl}(t+1)^{(\pl-1)(3-\pl)} k^{\pl^2 - 3\pl+1})$.
\end{corollary}

Let us use Lemma~\ref{lem:alpha12} to derive a stricter bound on the number of edges in a PLB graph with $1 < \pl < 2$.

\begin{lemma}
Let $G$ be a \dpl graph with parameters $\pl$ and $t$, where $1 < \pl < 2$.
Then, $G$ has $O(n^{2/\pl} (t+1)^{2 - 2/\pl})$ edges.
\end{lemma}

\fullversion{
\begin{proof}
By Lemma~\ref{lem:inc_edges}, the total number of edges of $G$ is at most $\sum_{i=1}^{n-1} d_{\geq i}$.
Let $\degth = n^{1/\pl}(t+1)^{1-1/\pl}$.
We split the sum into two parts.

By Lemma~\ref{lem:alpha12}, for $k \geq \degth$, $d_{\geq k} = O(n^{3-\pl}(t+1)^{(\pl-1)(3-\pl)} k^{\pl^2 - 3\pl + 1})$.
Moreover, observe that since $1 < \pl < 2$, $\pl^2 - 3\pl + 1 < -1$.
Then
\begin{align*}
\sum_{i=\degth}^{n-1} d_{\geq i} & = \sum_{i=\degth}^{n-1} O(n^{3-\pl}(t+1)^{(\pl-1)(3-\pl)} i^{\pl^2 - 3\pl + 1})\\
        & = O(n^{3-\pl}(t+1)^{(\pl-1)(3-\pl)} n^{1/\pl(\pl^2 - 3\pl + 2)}(t+1)^{(1-1/\pl)(\pl^2 - 3\pl + 2)}) \\
        & = O(n^{3-\pl + \pl - 3 + 2/\pl} (t+1)^{-\pl^2 + 4\pl - 3 + \pl^2 - 3\pl + 2 - \pl + 3 - 2/\pl})\\
        & = O(n^{2/\pl} (t+1)^{2 - 2/\pl}).
\end{align*}
On the other hand
\begin{align*}
\sum_{i=1}^{\degth} d_{\geq i} & = \sum_{i=1}^{\degth} O(n(t+1)^{\pl-1}i^{1-\pl})\\
        & = O(n(t+1)^{\pl-1}n^{1/\pl(2-\pl)} (t+1)^{(1-1/\pl)(2-\pl)})\\
        & = O(n^{2/\pl} (t+1)^{\pl-1 + 2 - \pl - 2/\pl + 1})\\
        & = O(n^{2/\pl} (t+1)^{2 - 2/\pl})
\end{align*}
Thus, $\sum_{i=1}^{n-1} d_{\geq i} = O(n^{2/\pl} (t+1)^{2 - 2/\pl})$.
\end{proof}
}

\fullversion{\subsection{\dpl neighborhoods}}
\exabstract{\paragraph{\dpl neighborhoods}}
Assume that we pick a random vertex from a power law graph with parameter $\pl$, proportionally to its degree.
Then, the degree of the chosen vertex comes from a power law distribution with parameter $\pl-1$.
This implies that, roughly speaking, for each vertex $v$ in a \emph{random} power law graph, the degree distribution of degrees of \emph{neighbors} of $v$ also obeys power law.
This fact can be exploited to obtain better running time bounds of some algorithms.
However, the algorithms that we later give actually rely on a weaker property.
Namely, for a vertex of degree $k$ they only need a bound on the number of neighbors of degree at least $k$.
Note that if we randomly pick $k$ vertices proportionally to their degrees, then the number of chosen vertices of degree at least $k$ is $O((t+1)^{\pl-2} k \sum_{i=k}^{n-1} i(i+t)^{-\pl})$.
This motivates the following.




\begin{definition}
\label{def:dpl_neigh}
Let $G$ be a \dpl{} graph with parameters $\pl > 2$ and $t$, and let $c_2 > 0$ be an universal constant.
We say that $G$ has \dpl{} neighborhoods if for every vertex $v$ of degree $k$, the number of neighbors of $v$ of degree at least $k$ is at most $c_2 \max(\log n, (t+1)^{\pl-2} k \sum_{i=k}^{n-1} i(i+t)^{-\pl})$.
\end{definition}

The $\log n$ factor in the definition comes from the fact that we assume that the graph is created in a random way.
Thus, the actual numbers of neighbors may slightly deviate from the expected values.

\begin{lemma}\label{lem:dpl_neigh}
Let $G$ be a \dpl{} graph with parameters $\pl > 2$ and $t$, and \dpl{} neighborhoods.
Then, for every vertex $v$ of degree $k$, the number of neighbors of $v$ of degree at least $k$ is $O(\max(\log n, (t+1)^{\pl-2} k^{3-\pl}))$.
\end{lemma}

\fullversion{
\begin{proof}
We have
\begin{align*}
    c_2 (t+1)^{\pl-2} k \sum_{i=k}^{n-1} i(i+t)^{-\pl} & \leq O((t+1)^{\pl-2} k \sum_{i=k}^{n-1} i^{1-\pl})\\
       &  = O((t+1)^{\pl-2} k\cdot k^{2-\pl})\\
       & = O((t+1)^{\pl-2} k^{3-\pl}).\\
\end{align*}
Thus, $c_2 \max(\log n, (t+1)^{\pl-2} k \sum_{i=k}^{n-1} i(i+t)^{-\pl}) = O(\max(\log n, (t+1)^{\pl-2} k^{3-\pl}))$.
\end{proof}
}

\fullversion{\subsection{Relation to other models}}
\exabstract{\paragraph{Relation to other models}}
Definitions~\ref{def:dpl} and~\ref{def:dpl_neigh} are designed to capture the properties of power law graphs that can be easily exploited in the analysis of algorithms.
At the same time there are many random graph models that produce power law graphs.
In these models even giving simple bounds on the degree distributions of the produced graphs is often highly nontrivial.
The analyses of some these models~\cite{ba, bollobas2001, kumar2000, grechnikov2014, Cooper01} only give the expected numbers of vertices of given degrees and analyze the concentration.
A typical concentration statement says that (with high probability) the number of vertices of degree $k$ differs from the expected value by some small additive error (e.g., $\sqrt{n} \log n$).
This cannot be directly used to show that these graphs satisfy Definition~\ref{def:dpl}.
Proving that would require bounding the number of vertices of degree belonging to $[2^d, 2^{d+1})$, but if we simply sum the approximate numbers of vertices of each degree $\in [2^d, 2^{d+1})$, the additive errors accumulate.
At the same time we believe that many of the proposed random graph processes yield PLB graphs, but proving this is a challenging open problem.

Another models for power law graphs are based on fixing a degree sequence in the beginning.
In the \emph{erased configuration model}~\cite{bender78,bollobas80,britton06} the degrees of all vertices are fixed in the very beginning to obtain an almost ideal power law distribution.
Then a graph is picked uniformly at random, among all graphs that have the given degree sequence.
Note that we fix an ``ideal'' power law degree sequence, as it is done, e.g., in~\cite{newman01,Aiello01,berry2014}, but in some works on this model the degree of each vertex is picked independently at random.

\begin{theorem}\label{thm:alpha-beta-plb}
Let $n$ be sufficiently large and $G$ be a random power law graph with parameter $\pl > 1$ created by erased configuration model.
Then $G$ is a \dpl graph with parameters $\pl$ and $t=0$.
Moreover, with high probability, $G$ has \dpl neighborhoods.
\end{theorem}

\fullversion{
These statements are true for some universal constants $c_1$ and $c_2$ (see definitions~\ref{def:dpl} and~\ref{def:dpl_neigh}).

The remaining part of this section gives a proof of Theorem~\ref{thm:alpha-beta-plb}.
Let us now describe the erased configuration model in detail.
First, we pick a degree for every vertex, in such a way that the number of vertices of degree $k$ is $\Theta(n / k^{\pl})$.
Since the sum of all degrees has to be even, we add one vertex of degree $1$ if necessary.
For simplicity of the analysis we ignore this added vertex.
Then, we build a random graph with the chosen degree sequence as follows:
\begin{enumerate}
\item Build a complete graph $H$ containing $\deg(v)$ copies of vertex $v$.
\item Choose a random perfect matching in $H$ and remove the edges that are not in the matching.
\item Build $G$ from $H$ by merging the copies of each vertex.
\end{enumerate}
The resulting graph may have multiple edges or self-loops, which we remove.

It follows easily that the maximum degree in $G$ is $O(n^{1/\pl})$.
We now verify that $G$ satisfies Definition~\ref{def:dpl}.
We have
\[
    \sum_{i=2^d}^{2^{d+1}-1} d_i \leq \sum_{i=2^d}^{2^{d+1}-1}O(n / i^{\pl}) = O(n) \sum_{i=2^d}^{2^{d+1}-1} i^{-\pl} \leq c_1 n \sum_{i=2^d}^{2^{d+1}-1} i^{-\pl}
\]
for some universal constant $c_1$.
Thus, $G$ is a \dpl graph with parameters $t=0$ and $\pl$.

The proof that, with high probability, $G$ has \dpl neighborhoods (satisfies Definition~\ref{def:dpl_neigh}) is more involved.
Let us now assume that $\pl > 2$ and fix a vertex $v$ of degree $k$.
Our goal is to bound the number of neighbors of $v$ of degree at least $k$.

Vertex $v$ has $k$ copies in $H$, that we denote by $v_1, \ldots, v_k$.
We say that a vertex of $H$ is \emph{bad} if it is a copy of a vertex of degree at least $k$, but not a copy of $v$.
Let us define a sequence of Boolean random variables $X_1, \ldots, X_k$, where $X_i = 1$ iff. $v_i$ is matched in $H$ with a bad vertex.
Note that matching $v_i$ with another copy of $v$ does no harm, as this creates a self loop in $G$, which is then removed.
Thus, $\sum_{i=1}^k X_i$ is an upper bound on number of neighbors of $v$ of degree at least $k$ (in $G$).
The number of bad vertices is bounded by
\[
\sum_{i=k}^{n-1} i O(n / i^\pl) \leq n \sum_{i=k}^{n-1} i^{1-\pl} = O(nk^{2-\pl}).
\]
Thus, $P(X_i = 1) \leq Ck^{2-\pl}$, for some universal constant $C$, as $P(X_i = 1)$ is bounded by the probability of a randomly chosen vertex being bad.
Define $X = \sum_{i=1}^k X_i$.
It follows that $\mathbb{E}(X) \leq Ck^{3-\pl}$.

We now use Chernoff bound to bound $X$.
The variables $X_i$ are not independent, but they are \emph{negatively associated}, which suffices for the Chernoff bound to work (see e.g.~\cite{panconesi}).

\begin{lemma}\label{lem:nega}
For any set $I \subseteq \{1, \ldots, k\}$, $\mathbb{P}(\bigwedge_{i \in I} X_i = 1) \leq \prod_{i \in I}\mathbb{P}(X_i = 1)$,
that is variables $X_i$ are negatively associated.
\end{lemma}

\begin{proof}
Assume the number of bad vertices is $b$.
Then $\mathbb{P}(X_i = 1) = b / (n - 1)$.

Observe that the perfect matching in $H$ can be computed as follows.
We go through the vertices in any order.
For each vertex, if it is already matched, we skip it.
Otherwise, we match it to a randomly chosen unmatched vertex.

For the purpose of the proof, we may assume that the first vertices that are chosen in this process are $v_1, \ldots, v_k$.
In the $i$-th step, we need to compute the probability that $v_i$ is matched to a bad vertex, provided that vertices $v_1, \ldots, v_{i-1}$ have been matched to a bad vertex.
This probability is clearly $(b-i+1) / (n-1-2i+2) \leq b / (n - 1)$.
Thus, $\mathbb{P}(\bigwedge_{i \in I} X_{i_1} = 1) \leq (b / (n - 1))^{|I|} = \prod_{i \in I}\mathbb{P}(X_i = 1)$.
The lemma follows.
\end{proof}

In our proof we use the following version of the Chernoff bound.
By Lemma~\ref{lem:nega} it can be applied to the random variables $X_1, \ldots, X_k$.

\begin{theorem}[\cite{panconesi}]\label{thm:chb}
Let $X = \sum_{i=1}^k X_i$, where $0 \leq X_i \leq 1$ and $X_i$ are negatively associated.
Let $t > 2e\mathbb{E}(X)$.
Then $\mathbb{P}(X > t) \leq 2^{-t}$.
\end{theorem}

We now proceed with the main part of the proof of the second claim of Theorem~\ref{thm:alpha-beta-plb}, which states that $G$ has \dpl neighborhoods (with high probability).
We show that the property of Definition~\ref{def:dpl_neigh} holds for a single vertex with high probability and then use union bound.

Set $t = \max(c\log n, 2e\mathbb{E}(X))$.
By Theorem~\ref{thm:chb}, we have $\mathbb{P}(X > t) \leq 2^{-t} \leq n^{-c}$.
Thus, with high probability a vertex with degree $k$ has at most $t$ neighbors.
We now show that $t \leq c_2(\max(\log n, k \sum_{i=k}^{n-1} i^{1-\pl}))$ for some universal constant $c_2$.
The case when $t = c \log n$ is trivial, as we may set $c_2 = c$.
Now, assume that $t = (\log n)2e\mathbb{E}(X)$.
We have that $\mathbb{E}(X) \leq Ck^{3-\pl}$.
By Lemma~\ref{lem:sum-rev}, $k^{2-\pl} \leq C'\sum_{i=k}^{n-1} i^{1-\pl}$ for some universal constant $C'$.
The Lemma requires that $k \leq (n-1)/2$, which follows from the fact that the degrees are bounded by $O(n^{1/\pl})$.
Hence,
\[
t = 2e\mathbb{E}(X) \leq 2e Ck^{3-\beta} = O(k \sum_{i=k}^{n-1} i^{1-\pl}),
\]
which completes the proof of Theorem~\ref{thm:alpha-beta-plb}.
}

\section{Real-World Networks are Power-Law Bounded}
\label{sec:def_experiments}
In this section we verify our definitions from Section~\ref{sec:plb} on real-world networks.
The majority of our graphs comes from Stanford Large Network Dataset Collection~\cite{stanford}.
In addition, we analyze the global flights network~\cite{openflights}, as well as WIW social network degree distribution~\cite{PhysRevE.69.036131}.\footnote{We thank the authors of~\cite{PhysRevE.69.036131} for sharing with us this data.}

First, we focus on Definition~\ref{def:dpl}.
We compute the degree distributions of each network and then try to choose the parameters $c_1, \pl$ and $t$,
so that our bound on the number of vertices of given degree is as tight as possible.
At the same time we ensure that $c_1$ is at most $5$ (as it is supposed to be a constant) and try to maximize $\pl$,
since larger $\pl$ implies better running time bounds of our algorithms.
The results of this adjustments are shown in Table~\ref{table:plb_adjustment}.
Observe that the value of $t$ is very small compared to $n$.
Some of the graphs in the data sets are directed.
For such graphs we make two adjustments.
Either we drop the orientations of the edges (``(directed, in-degree + out-degree)'' in Table~\ref{table:plb_adjustment}), or we slightly modify Definition~\ref{def:dpl} and only consider the outdegrees of vertices (``(directed, out-degree)'' in Table~\ref{table:plb_adjustment}).
For some of the networks, in Fig.~\ref{fig:dpl_graphs} we also show the degree distribution, as well as the bound of Definition~\ref{def:dpl}.
In order to show the data with more detail, we plot not only the numbers of vertices, whose degrees belong to $[2^d, 2^{d+1})$ (actual and the upper bounds of Definition~\ref{def:dpl}), but also the numbers of vertices of degree belonging to $[k, 2k)$ for each $1 \leq k < n$.

In the case when $\pl<2$, in Fig.~\ref{fig:dpl_graphs} the bound from Corollary~\ref{corollary-double-power-law} is marked with green line.
For Epinions and WikiTalk graphs the critical degree, when the second power law starts, is predicted rather well. Note that this is not the case
for Facebook graph as the maximum number of friends one can have is limited to $5000$. The high degree part of the distribution is cut off at this number.
While the critical degree is predicted decently, the slope of the second power law distribution is underestimated. This is most probably due to the
worst-case form of our bounds which are overly pessimistic with respect to the actual trend.

Then, we move on to Definition~\ref{def:dpl_neigh}.
For each network we use the previously computed parameters $\pl$ and $t$ and find the smallest value of $c_2$, for which the definition is satisfied.
We skip the graphs, where $\pl < 2$, as Definition~\ref{def:dpl_neigh} does not apply to them.
The values of $c_2$ obtained this way are also shown in Table~\ref{table:plb_adjustment}.
Observe that for every network the computed value of $c_2$ is less than $8.06$, and for a big majority of them it is less than $2$.
This confirms that the property of Definition~\ref{def:dpl_neigh} is indeed present in real-world graphs.

Table~\ref{table:plb_adjustment} also contains two adjustments, in which we force the value of $t$ to be $0$.
In some sense this is similar to fitting the standard definition of a power law distribution to our data.
However, this causes the value of $\pl$ to increase and makes our bounds much further from the real data, as shown in Fig~\ref{fig:dpl_graphs_t0}.

\section{Counting Triangles and Maximal Clique}\label{sec:algorithms}
This section presents our first two algorithms for PLB graphs.
The first algorithm counts triangles, whereas the second one returns the size of the maximal clique.
The algorithms themselves are easy and should be considered folklore.
However, we show that in the case of PLB graphs they perform much better than in the case of general graphs.
Then we obtain even better running time bounds for graphs with \dpl neighborhoods.
This is the most important contribution of this section, as we believe that it gives the first solid explanation of the good performance of triangle counting and maximum clique algorithms in real-world graphs.

Both our algorithms are based on the same construction.
We first direct the edges of $G$ towards vertices of higher degree.
Formally, let $v_1, \ldots, v_n$ be all vertices of $G$ sorted in non-decreasing order of degrees.
We define $\dir{G}$ to be a graph obtained from $G$ by directing each undirected edge $v_iv_j$ towards $v_{\max(i,j)}$.
Note that since the degrees of vertices are bounded by the number of vertices, we may sort the vertices and build $\dir{G}$ in linear time.
Moreover, note that $\dir{G}$ does not contain any cycles.
Let $b(k)$ be the maximum out-degree in $\dir{G}$ of a vertex of degree $k$ in $G$.

Note that the value $b(k)$ is related to \emph{graph degeneracy}.
We say that a graph is $d$-degenerate if every subgraph has a vertex of degree at most $d$.
In our case $G$ is $d$-degenerate for $d = \max_{i=1, \ldots, n-1} b(i)$.
In $d$-degenerate graphs we can count triangles in $O(dm)$ time~\cite{chiba85}.
Since \dpl graphs are $O(n^{1/\pl})$-degenerate (assuming $t=0$), this can be used to obtain a running time bound of $O(mn^{1/\pl})$, which is the same as the running time given in~\cite{Latapy2008458}.
However, with a slightly more careful analysis, in this section we improve this bound.
While this result is simple, to the best of our knowledge it has not been previously stated explicitly.

We first use the bounds derived in Section~\ref{sec:plb} to bound $b(k)$.
\begin{lemma}
\label{lem:directed_edges}
Let $G$ be a \dpl{} graph with parameters $\pl$ and $t$.
Then $b(k) = O(\min(k, d_{\geq k})) = O(\min(k, n(t+1)^{\pl-1}k^{1-\pl}) = O(n^{1/\pl}(t+1)^{1-1/\pl}).$
\end{lemma}

\fullversion{
\begin{proof}
Obviously $b(k) \leq k$, since for every $v \in V(G)$, we have $\outdeg_{\dir{G}}(v) \leq \deg_G(v)$.
In addition to that, since the edges are directed towards vertices of higher degree, $b(k) \leq d_{\geq k}$.
The second inequality follows directly from the bound on $d_{\geq k}$ derived in Lemma~\ref{lem:high_degree_num}.

It remains to show that $O(\min(k, n(t+1)^{\pl-1}k^{1-\pl})) = O(n^{1/\pl}(t+1)^{1-1/\pl})$.
Assume that $k \geq n^{1/\pl}(t+1)^{1-1/\pl}$. Then
\[
n(t+1)^{\pl-1}k^{1-\pl} \leq n(t+1)^{\pl-1}n^{1/\pl-1}(t+1)^{2-\pl-1/\pl} = n^{1/\pl} (t+1)^{1-1/\pl},
\]
as desired.
The lemma follows.
\end{proof}
}

\fullversion{If our graph additionally has \dpl{} neighborhoods (see Definition~\ref{def:dpl_neigh}), we may obtain a better bound.}

\begin{lemma}
\label{lem:directed_edges_dpl}
Let $G$ be a \dpl{} graph with parameters $\pl > 2$ and $t$, and \dpl{} neighborhoods.
Then $b(k) = O(\min(n(t+1)^{\pl-1}k^{1-\pl}, \log n + (t+1)^{\pl-2}k^{3-\pl})) = O(\log n + (t+1)^{\pl/2 - 1/2}n^{3/2-\pl/2})$
\end{lemma}

\fullversion{
\begin{proof}
Observe that $b(k)$ is at most the number of neighbors of degree at least $k$ in the neighborhood.
Thus, by Lemma~\ref{lem:dpl_neigh}, $b(k) = O(\max(\log n, (t+1)^{\pl-2}k^{3-\pl}))$.
By Lemma~\ref{lem:directed_edges} we also have $b(k) = O(n(t+1)^{\pl-1}k^{1-\pl})$.
Thus, we get $b(k) = O(\min(n(t+1)^{\pl-1}k^{1-\pl}, \log n + (t+1)^{\pl-2}k^{3-\pl})$.

To balance the terms we take $k = \sqrt{(t+1)n}$.
Thus, $O(\min(n(t+1)^{\pl-1}k^{1-\pl}, \log n + (t+1)^{\pl-2}k^{3-\pl})) = O(\log n + (t+1)^{\pl-2}((t+1)n)^{3/2-\pl/2}) = O(\log n + (t+1)^{\pl/2 - 1/2}n^{3/2-\pl/2})$.
\end{proof}
}

\fullversion{\subsection{Counting Triangles}}
\exabstract{\paragraph{Counting Triangles}}
We now show efficient algorithms for counting triangles in an undirected \dpl{} graph $G$ with parameter $\pl$.
Their pseudocodes are given as Algorithms~\ref{alg:triangles_so} and~\ref{alg:triangles}.
The first algorithm is clearly structure-oblivious.
The second one also does not use the structure of the graph explicitly.
However, it takes a parameter $\degth$, which will depend on graph parameters $\pl$ and $t$.
\fullversion{Observe that Algorithm~\ref{alg:triangles_so} can easily be extended to list triangles in the same running time bound.}
\setlength{\textfloatsep}{10pt}
\begin{algorithm}
\begin{algorithmic}[1]
\Function{\ctr}{$G$}
\State{Construct $\dir{G}$}
\State{$triangles := 0$}
\For{$v \in V(G)$}
    \State{$S := $ set of endpoints of outedges of $v$}
    \For{each inedge $wv$ of $v$ in $\dir{G}$}\label{l:eachedge}
        \For{each outedge $wu$ of $w$ in $\dir{G}$}\label{l:inner_for}
            \If{$u \in S$}\label{l:eachpair}
                \State{$triangles := triangles + 1$}
            \EndIf
        \EndFor
    \EndFor
\EndFor
\Return{$triangles$}
\EndFunction
\end{algorithmic}
\caption{\label{alg:triangles_so}Structure-oblivious algorithm for counting triangles}
\end{algorithm}
\begin{algorithm}
\begin{algorithmic}[1]
\Function{\ctrfmm}{$G$, $\degth$}
\State{Construct $\dir{G}$}
\State{$triangles := 0$}
\For{$v \in V(G)$}\label{l:fs}
    \State{$S := $ set of endpoints of outedges of $v$}
    \For{each inedge $wv$ of $v$ in $\dir{G}$}
        \If{$\deg_G(w) \leq \degth$}
            \For{each outedge $wu$ of $w$ in $\dir{G}$}
                \If{$u \in S$}
                    \State{$triangles := triangles + 1$}
                \EndIf
            \EndFor
        \EndIf
    \EndFor
\EndFor\label{l:efs}
\State{$G_\degth := $ subgraph of $G$ induced on vertices of degree more than $\degth$}
\State{\Return{$triangles\;+ $ the number of triangles in $G_\degth$, counted using fast matrix multiplication}}
\EndFunction
\end{algorithmic}
\caption{\label{alg:triangles}Algebraic algorithm for counting triangles}
\end{algorithm}
\exabstract{\vspace{-10pt}}
\begin{lemma}\label{lem:ctr}
Algorithms \ctr and \ctrfmm are correct.
Their running times are $O(\sum_{i=1}^{n-1} d_i b(i)^2)$, and $O(\sum_{i=1}^{\degth} d_i b(i)^2 + d_{\geq \degth}^{\omega})$, respectively.
\end{lemma}

\fullversion{
\begin{proof}
Let us first consider the running time of \ctr.
Observe that the body of the \textbf{for} loop in the~\ref{l:eachedge}th line is ran exactly once per each edge of $\dir{G}$.
Thus, the \textbf{for} loop in the~\ref{l:inner_for} line is ran at most $b(\deg(w))$ times for a vertex $w$.
In other words, the~\ref{l:eachpair} line is executed for each pair of vertices $u$ and $v$, which are endpoints of the outedges of $w$.
This requires $O(\sum_{i=1}^{n-1} d_i b(i)^2)$ time.

The set $S$ can be implemented as a Boolean array.
This way we can initialize the set each time in linear time.
Moreover, we can test for membership in constant time.
Moreover, as observed before $\dir{G}$ can be computed in linear time.
Thus, \ctr runs in $O(\sum_{i=1}^{n-1} d_i b(i)^2)$ time.

Concerning correctness, let $v_1, \ldots, v_n$ be all vertices of $G$ sorted in non-decreasing order of degrees.
Consider a triangle $T$.
Let $v$ is the vertex of $T$ that comes first in the sorted order $v_1, \ldots, v_n$.
Then, in $\dir{G}$ the two edges of $T$ that are incident to $v$ are out-edges of $v$.
Thus, the correctness of \ctr follows.

Using similar arguments, we may observe that the first stage of \ctrfmm (lines~\ref{l:fs}--\ref{l:efs}) correctly identifies exactly the triangles that contain at least one vertex of degree at most $\degth$.
Clearly, $G_\degth$ contains exactly the triangles that have not been identified yet.
Since $G_\degth$ has exactly $d_{\geq \degth}$ vertices, the running time of \ctrfmm follows.
\end{proof}
}

\exabstract{We now combine the algorithms with the bounds on $b(k)$ derived in Lemmas~\ref{lem:directed_edges} and~\ref{lem:directed_edges_dpl} to obtain the following four results.
The running times that we derive below are shown in Fig.~\ref{fig:times}.
}

\fullversion{We now combine the algorithms with the bounds on $b(k)$ derived in lemmas~\ref{lem:directed_edges} and~\ref{lem:directed_edges_dpl} to obtain four running time bounds of our algorithms, that depend on the algorithm used and on whether the graph has \dpl neighborhoods.
These running times are shown in Fig.~\ref{fig:times}.}

\begin{theorem}
\label{lem:listing_triangles}
Let $G=(V,E)$, $n = |V|$ be a \dpl{} graph with parameters $\pl$ and $t$.
Then, algorithm \ctr can compute the number of triangles in $G$ in: (a) $O(n^{3/\pl} (t+1)^{3-3/\pl})$ time for $1 < \pl < 3$, (b) $O(n \log n (t+1)^2)$ time for $\pl = 3$, (c) $O(n (t+1)^2)$ time for $\pl > 3$.
\end{theorem}

\fullversion{
\begin{proof}
By Lemma~\ref{lem:ctr}, the running time is $O(\sum_{i=1}^{n-1} d_i b(i)^2)$.
Let $1 \leq \degth < n$ be a parameter that we fix later.
We split the sum into two pieces and first bound $\sum_{i=1}^\degth d_i O(b(i)^2)$.
By Lemma~\ref{lem:rtime_sum} this can be upper bounded by

\[
O(1)n(t+1)^{\pl-1} \sum_{i=1}^\degth (i+t)^{-\pl}i^2
\]

There are now three cases to consider, depending on the value of $\pl$.
For $\pl > 3$ we have
\[
    \sum_{i=1}^\degth (i+t)^{-\pl} i^2 \leq \sum_{i=1}^t (i+t)^{-\pl} i^2 + \sum_{i=t+1}^\degth i^{2-\pl} \leq (t+1)^{-\pl} \sum_{i=1}^t i^2 + \sum_{i=t+1}^\degth i^{2-\pl} = O((t+1)^{3-\pl}),
\]
so the running time is $O(1)n(t+1)^{\pl-1} (t+1)^{3-\pl} = O((t+1)^2 n)$, regardless of the choice of $\degth$.
Thus, we may set $\degth = n-1$.

For $\pl \leq 3$, we have
\[
    O(1)n(t+1)^{\pl-1} \sum_{i=1}^\degth (i+t)^{-\pl}i^2 \leq O(1)n(t+1)^{\pl-1} \sum_{i=1}^\degth i^{2-\pl}.
\]

For $\pl = 3$ this gives $O(n\log(\degth)(t+1)^2)$.
Again, we set $\degth = n-1$ and obtain a running time of $O(n \log n (t+1)^2)$.
The last case is when $1 < \pl < 3$.
Then, the sum is equal to $O(n(t+1)^{\pl-1} \degth^{3-\pl})$.

In this case we set $\degth < n-1$, so we still need to bound $\sum_{i=\degth+1}^{n-1} d_i O(b(i)^2)$.
We use the fact that $b(i) = O(d_{\geq i})$ (see Lemma~\ref{lem:directed_edges}) and $d_{\geq \degth} = O(n(t+1)^{\pl-1}\degth^{1-\pl})$ (see Lemma~\ref{lem:high_degree_num}):
    \[
        \sum_{i=\degth+1}^{n-1} d_i O(b(i)^2) = \sum_{i=\degth+1}^{n-1} d_i d_{\geq i}^2 \leq \sum_{i=\degth+1}^{n-1} d_i d_{\geq \degth}^2 = O(d_{\geq \degth}^3) = O(n^3(t+1)^{3(\pl-1)} \degth^{3(1-\pl)})
    \]
The overall running time is $O(n(t+1)^{\pl-1}\degth^{3-\pl} + n^3(t+1)^{3(\pl-1)} \degth^{3(1-\pl)})$.
In order to balance the summands, we set $\degth = n^{1/\pl} (t+1)^{1-1/\pl}$ and obtain the running time of $O(n^{3/\pl} (t+1)^{3-3/\pl})$.
\end{proof}
}


\begin{theorem}\label{thm:listing_triangles_fmm}
Let $G$ be a \dpl{} graph with parameters $\pl$ and $t$, where $\pl < 3$.
Then algorithm \ctrfmm with $\degth = (n(t+1)^{\pl-1})^{(\omega-1)/(3-\pl-\omega+\pl\omega)}$ can compute the number of triangles in $G$ in $O((n(t+1)^{\pl-1})^{3.45/(0.45+\pl)})$ time.
\end{theorem}

\fullversion{
\begin{proof}
By Lemma~\ref{lem:ctr}, the running time is $O(\sum_{i=1}^{\degth} d_i b(i)^2 + d_{\geq \degth}^{\omega})$.
In the proof of Theorem~\ref{lem:listing_triangles} we have shown that for $1 < \pl < 3$, $O(\sum_{i=1}^{\degth} d_i b(i)^2) = O(n(t+1)^{\pl-1}\degth^{3-\pl})$.
By Lemma~\ref{lem:high_degree_num}, $d_{\geq \degth} = O(n(t+1)^{\pl-1}\degth^{1-\pl})$
The running time becomes $O(n(t+1)^{\pl-1}\degth^{3-\pl} + n^\omega(t+1)^{\omega(\pl-1)} \degth^{\omega(1-\pl)})$.
We balance both summands:
\begin{align*}
n(t+1)^{\pl-1}\degth^{3-\pl} & = n^\omega(t+1)^{\omega(\pl-1)} \degth^{\omega(1-\pl)}\\
\degth^{3-\pl-\omega+\pl\omega} & = n^{\omega-1}(t+1)^{(\omega-1)(\pl-1)}\\
\degth & = (n(t+1)^{\pl-1})^{(\omega-1)/(3-\pl-\omega+\pl\omega)}\\
\end{align*}
and obtain a running time of
\begin{align*}
O(n(t+1)^{\pl-1}\degth^{3-\pl}) & = O((n(t+1)^{\pl-1})(n(t+1)^{\pl-1})^{(3-\pl)(\omega-1)/(3-\pl-\omega+\pl\omega)})\\
& = O((n(t+1)^{\pl-1})^{(3-\pl)(\omega-1)/(3-\pl-\omega+\pl\omega)+1}) \\
& = O((n(t+1)^{\pl-1})^{2\omega/(3-\pl-\omega+\pl\omega)})
\end{align*}
Setting $\omega = 2.38$, we get
\begin{align*}
O((n(t+1)^{\pl-1})^{4.76/(0.62+1.38\pl)}) = O((n(t+1)^{\pl-1})^{3.45/(0.45+\pl)}).
\end{align*}

For $\pl = 2$, $3.45/(0.45+\pl) \leq 1.41$, so the running time becomes $O(n^{1.41}(t+1)^{1.41})$.
\end{proof}
}
\exabstract{Note that for $\pl = 2$, the running time becomes $O(n^{1.41}(t+1)^{1.41})$.}

\fullversion{If $G$ additionally has \dpl{} neighborhoods, we may obtain a faster algorithm.}

\begin{theorem}
\label{thm:triangles_dpl_neigh}
Let $G$ be a \dpl{} graph with parameters $\pl$ and $t$, where $\pl < 3$.
Moreover, assume that $G$ has \dpl{} neighborhoods.
Then, algorithm \ctr can compute the number of triangles in $G$ in time (a) $O(n^{9/2-3/2\pl}(t+1)^{3/2\pl - 3/2})$ for $2 < \pl < 7/3$,
(b) $\tilde{O}(n (t+1)^2)$ time for $\pl \geq 7/3$.
\end{theorem}

\fullversion{
\begin{proof}
By Lemma~\ref{lem:ctr} the running time is $O(\sum_{i=1}^{n-1} d_i b(i)^2)$.
By Lemma~\ref{lem:directed_edges_dpl}, $b(k) = O(\min(n(t+1)^{\pl-1}k^{1-\pl}, \log n + (t+1)^{\pl-2}k^{3-\pl}))$.
In particular, $b(k) = O(\log n + (t+1)^{\pl-2}k^{3-\pl})$.

Again, split the sum using a parameter $\degth$ that we fix later.

We first bound the sum of the first $\degth$ summands (we use Lemma~\ref{lem:rtime_sum}):
\begin{align*}
\sum_{i=1}^{\degth} d_i b(i)^2 & = O(1) n (t+1)^{\pl-1} \sum_{i=1}^\degth b(i)^2 (i+t)^{-\pl}\\
 & = O(n (t+1)^{\pl-1}\sum_{i=1}^\degth (\log^2 n + (t+1)^{2\pl-4}i^{6-2\pl}) (i+t)^{-\pl})\\
 & = O(n (t+1)^{3\pl-5} (\sum_{i=1}^\degth \log^2 n (i+t)^{-\pl}) + (\sum_{i=1}^\degth i^{6-2\pl} (i+t)^{-\pl}))\\
 & = O(n (t+1)^{3\pl-5} ((t+1)^{1-\pl} \log^2 n  + (\sum_{i=1}^\degth i^{6-2\pl} (i+t)^{-\pl}))\\
 & = O(n\log^2 n (t+1)^{2\pl-4} + n (t+1)^{3\pl-5} \sum_{i=1}^\degth i^{6-2\pl} (i+t)^{-\pl})\\
\end{align*}
We use Lemma~\ref{lem:sum-t}.
If $\pl \geq 7/3$, the sum is bounded by $O((\delta+1)^{7-3\pl})$.
Thus, if we set $\degth = n-1$, the running time becomes $\tilde{O}(n (t+1)^2)$.

It remains to consider the case when $2 < \pl < 7/3$.
Then, we assume that $\degth = n^{\Omega(1)}$ (which we can do, since we are free to choose $\degth$).
The sum can be bounded by $O(\degth^{7-3\pl})$, so $\sum_{i=1}^{\degth} d_i b(i)^2 = O(n(t+1)^{3\pl-5}\degth^{7-3\pl})$.

To obtain the running time we still need to bound $\sum_{i=\degth+1}^{n-1} d_i b(i)^2$.
In the proof of Theorem~\ref{lem:listing_triangles}, we have shown that $\sum_{i=\degth+1}^{n-1} d_i b(i)^2 = O(n^3(t+1)^{3(\pl-1)}\degth^{3(1-\pl)})$.
We balance $n (t+1)^{3\pl-5} \degth^{7-3\pl}$ and $n^3(t+1)^{3(\pl-1)}\degth^{3(1-\pl)}$:
\begin{align*}
n (t+1)^{3\pl-5} \degth^{7-3\pl} & = n^3(t+1)^{3(\pl-1)}\degth^{3(1-\pl)}\\
 \degth^4 & = n^2(t+1)^{2}\\
 \degth & = \sqrt{n(t+1)}
\end{align*}
The running time becomes $O(n^{9/2-3/2\pl}(t+1)^{3/2\pl - 3/2})$.
\end{proof}
}
Thanks to Theorem~\ref{thm:alpha-beta-plb}, Theorem~\ref{thm:triangles_dpl_neigh} applies (whp) to graphs generated by erased configuration model.
Thus, this algorithm generalizes and strengthens the result of Berry et al.~\cite{berry2014} by showing whp bounds on the running time 
instead of bounds in expectation. 

\begin{theorem}
Let $G$ be a \dpl{} graph with parameters $2 < \pl < 7/3$ and $t$.
Moreover, assume that $G$ has \dpl{} neighborhoods.
Then algorithm \ctrfmm for $\degth = n^{(\omega-1)/(7-\omega+(\omega-3)\pl)} (t+1)^{1 - 2/(7-\omega+(\omega-3)\pl)}$ can compute the number of triangles in $G$ in $O(n^{(23.04 - 7.68\pl)/(7.46-\pl)})(t+1)^{(7.67(\pl-1))/(7.46-\pl)})$ time.
\end{theorem}

\fullversion{
\begin{proof}
By Lemma~\ref{lem:ctr} the running time is $O(\sum_{i=1}^{\degth} d_i b(i)^2 + d_{\geq \degth}^{\omega})$.
In the proof of Theorem~\ref{thm:triangles_dpl_neigh} we have shown that for $2 < \pl < 7/3$ and $\degth = n^{\Omega(1)}$, $\sum_{i=1}^{\degth} d_i b(i)^2 = O(n (t+1)^{3\pl-5} \degth^{7-3\pl})$.
On the other hand, as in the proof of Theorem~\ref{thm:listing_triangles_fmm}, it takes $O(n^\omega(t+1)^{\omega(\pl-1)} \degth^{\omega(1-\pl)})$ to process vertices of degree at least $\degth$.
We balance both times to find the optimal choice for $\degth$.
\begin{align*}
n (t+1)^{3\pl-5} \degth^{7-3\pl} & = n^\omega(t+1)^{\omega(\pl-1)} \degth^{\omega(1-\pl)}\\
      \degth^{7-\omega+(\omega-3)\pl} & = n^{\omega-1} (t+1)^{(\omega-3)\pl - \omega + 5}\\
      \degth & = n^{(\omega-1)/(7-\omega+(\omega-3)\pl)} (t+1)^{1 - 2/(7-\omega+(\omega-3)\pl)}
\end{align*}
By plugging this back, we obtain the running time of
\begin{align*}
O(n^\omega(t+1)^{\omega(\pl-1)}& n^{((\omega-1)\omega(1-\pl))/(7-\omega+(\omega-3)\pl)} (t+1)^{(1 - 2/(7-\omega+(\omega-3)\pl))\omega(1-\pl)})\\
       &  = O(n^{\omega(\omega-\omega\pl-1+\pl+7-\omega+(\omega-3)\pl)/(7-\omega+(\omega-3)\pl)} (t+1)^{(2\omega(\pl-1))/(7-\omega+(\omega-3)\pl)})\\
        & = O(n^{(2\omega(3-\pl))/(7-\omega+(\omega-3)\pl)}(t+1)^{(2\omega(\pl-1))/(7-\omega+(\omega-3)\pl)})
\end{align*}
Setting $\omega=2.38$, we get $O(n^{(14.28-4.76\pl)/(4.62-0.62\pl)}(t+1)^{(4.76(\pl-1))/(4.62-0.62\pl)}) = O(n^{(23.04 - 7.68\pl)/(7.46-\pl)}(t+1)^{(7.67(\pl-1))/(7.46-\pl)})$.
For $\pl \to 2_{+}$ this becomes $O(n^{1.41}(t+1)^{1.41})$.
For $\pl \to 7/3$, it is $O(n(t+1)^2)$.
\end{proof}
}
\exabstract{Note that for $\pl \to 2_{+}$, the running time is $O(n^{1.41}(t+1)^{1.41})$, whereas for $\pl \to 7/3$, it is $O(n(t+1)^2)$.}

\fullversion{\subsection{Finding Maximal Clique}}
\exabstract{\paragraph{Finding Maximal Clique}}
We now show an efficient algorithm for finding the largest clique in \dpl{} graph.

\begin{algorithm}
\begin{algorithmic}[1]
\Function{\textsc{MaximalClique}}{$G$}
\State{Construct $\dir{G}$}
\State{$maxclique := 0$}
\For{$v \in V(G)$}
    \State $N_v := $ $\{v\} \; \cup $ set of endpoints of outedges of $v$ in $\dir{G}$
    \For{$S \subseteq N_v$}
        \If{$S$ is a clique in $G$}
            $maxclique := \max(maxclique, |S|)$
        \EndIf
    \EndFor
\EndFor
\Return{$maxclique$}
\EndFunction
\end{algorithmic}
\caption{Maximal clique algorithm}
\end{algorithm}
\exabstract{\vspace{-10pt}}

\begin{lemma}\label{lem:clique_correct}
Algorithm \textsc{MaximalClique} is correct, structure oblivious and runs in $n\sum_{i=1}^{n-1} \exp(b(i))$ time.
\end{lemma}

\fullversion{
\begin{proof}
Let $C$ be a clique in $G$.
Then, $C$ contains a vertex $w$ such that in $\dir{G}$ $w$ has directed edges to every other vertex of $C$.
The correctness of the algorithm follows easily.
It is also easy to see that it is structure-oblivious.

Consider the iteration of the outer \textbf{for} loop for a vertex $v$.
The size of $N_v$ is bounded by $b(\deg(v))+1$, so the inner \textbf{for} loop
runs in $\exp(b(\deg(v)))$ time.
For a single $v$ this can be crudely upper bounded by $\sum_{i=1}^{n-1} exp(b(i))$.
The outer loop has $n$ iterations, so the entire algorithm runs in $n\sum_{i=1}^{n-1} \exp(b(i))$ time.
\end{proof}
}

\exabstract{Again, we use the bounds from Lemmas~\ref{lem:directed_edges} and~\ref{lem:directed_edges_dpl} to obtain two bounds on the running time.}

\begin{theorem}
Let $G$ be a \dpl{} graph with parameters $\pl$ and $t$.
Then, algorithm \textsc{MaximalClique} can find the largest clique in $G$ in $\exp(O(n^{1/\pl}(t+1)^{1-1/\pl}))$ time.
\end{theorem}

\fullversion{
\begin{proof}
By Lemma~\ref{lem:clique_correct}, the running time is $\poly(n)\sum_{i=1}^{n-1} \exp(b(i))$.
By Lemma~\ref{lem:directed_edges}, $b(k) = O(n^{1/\pl}(t+1)^{1-1/\pl})$.
Thus,
\[
    n\sum_{i=1}^{n-1} \exp(b(i)) = n \exp(O(n^{1/\pl}(t+1)^{1-1/\pl})) = \exp(O(n^{1/\pl}(t+1)^{1-1/\pl})).
\]
\end{proof}
}

\fullversion{This problem can also be solved more efficiently for a \dpl{} graph with parameter $\pl > 2$ and \dpl{} neighborhoods.}

\begin{theorem}
Let $G$ be \dpl{} graph with parameters $\pl > 2$ and $t$ and \dpl{} neighborhoods.
Then, algorithm \textsc{MaximalClique} can find the largest clique in $G$ in (a) $\exp(O((t+1)^{\pl/2-1/2}n^{3/2 - \pl/2}\log n))$ time for $2 < \pl < 3$, (b) $n^{O(t+1)}$ time for $\pl = 3$, (c) $O(\poly(n))$ time for $\pl > 3$.
\end{theorem}

\fullversion{
\begin{proof}
By Lemma~\ref{lem:clique_correct}, the running time is $\poly(n)\sum_{i=1}^{n-1} exp(b(i))$.
Moreover, by Lemma~\ref{lem:directed_edges_dpl}, $b(k) = O(\log n((t+1)^{\pl/2-1/2}n^{3/2 - \pl/2}+1))$.
Thus,
\[
    n \sum_{i=1}^{n-1} \exp(b(i)) = n^2 \exp(O(\log n((t+1)^{\pl/2-1/2}n^{3/2 - \pl/2})+1)) = \exp(O(\log n((t+1)^{\pl/2-1/2}n^{3/2 - \pl/2})+1)).
\]
If $\pl < 3$ this can be simplified to $\exp(O((t+1)^{\pl/2-1/2}n^{3/2 - \pl/2} \log n))$.
For $\pl = 3$, $n^{3/2 - \pl/2} = O(1)$, so we get $n^{O(t+1)}$.
For $\pl > 3$, we use the fact that $t+1 = O(n^{\epsilon})$ for every $\epsilon > 0$.
Thus, $(t+1)^{\pl/2-1/2}n^{3/2 - \pl/2} = O(1)$, so the running time becomes $\exp(O(\log n)) = \poly(n)$.

\end{proof}
}

Observe that for $\pl > 3$ the running time is \emph{polynomial} in $n$. Note that the analysis assumes $t = O(n^{\epsilon})$.

 \section{Algebraic Algorithms}\label{sec:algebraic}
\fullversion{
In this section we will give our algebraic algorithms for computing
matrix determinant and solving linear systems of equations. As already mentioned
we will be working over a finite field $\mathcal{F}$. For a warm-up we
will start from the generic symmetric case and next
we move on to general non-symmetric case.
}\exabstract{
In this section we will give our algebraic algorithms for computing
matrix determinant and solving linear systems of equations.
}In this section when we talk about directed \dpl graphs we assume that
only the outdegrees of vertices satisfy a similar bound to the one given in Definition~\ref{def:dpl}.
Moreover, we will use fast rectangular matrix multiplication. We denote by $\omega(n,m,k)$ the time needed
to multiply an $n\times m$ matrix by an $m\times k$ matrix~\cite{6375330}.

\fullversion{\subsection{Transitive Closure}}
\exabstract{\paragraph{Transitive Closure}}
Let us start by giving our algorithms for computing
transitive closure of a graph.

\begin{theorem}
\label{thm-transitive-closure}
Let $G$ be a directed \dpl{} graph with parameters $\pl$ and $t$, and let $1\le k< n^{1/\pl}(t+1)^{1-1/\pl}$.
Then, we can compute the the transitive closure of $G$ in $O(n^2(t+1)^{\pl-1} k^{2-\pl}+\omega(n,n,n(t+1)^{\pl-1}k^{1-\pl}))$ time.
\end{theorem}
\begin{proof}

Let $M$ be the adjacency matrix of $G$. We start
with sorting the rows of $M$ in decreasing order according to the number of non-zero entries.
Then, for a given $k \in [0,n]$, we split the matrix $M$ into $4$ submatrices

\begin{equation*}
M = \begin{bmatrix}
A & B \\
C & D
\end{bmatrix}
\end{equation*}
where $\begin{bmatrix} A & B \end{bmatrix}$ contains rows with more than $k$ non-zero entries, $\begin{bmatrix} C & D \end{bmatrix}$ has at most $k$ non-zero entries in each row and $A$, $D$ are square matrices.
Let $m_{CD}$ be the total number of non-zero entries in submatrices $C, D$.
By Lemma~\ref{lem:small_degrees} $m_{CD}$ is bounded $O(n(t+1)^{\pl-1} k^{2-\pl})$.
Let $n_k$ be the dimension of $A$. By Lemma~\ref{lem:high_degree_num} we have that $n_k$ is bounded by $O(n(t+1)^{\pl-1}k^{1-\pl})$.

We can express the
transitive closure of $M$ in the following block form
\begin{equation}
\label{eqn-tc}
M^*=
\left[\begin{matrix}
A & B\\
C & D
\end{matrix}
\right]^{*}=\left[\begin{matrix}
I & 0\\
-D^{*}C & I
\end{matrix}
\right]
\left[\begin{matrix}
(A-BD^{*}C)^{*} & 0\\
0 & D^{*}
\end{matrix}
\right]
\left[\begin{matrix}
I & -BD^{*}\\
0 & I
\end{matrix}
\right].
\end{equation}
\exabstract{
In order to compute the transitive closure using this equation we compute:
$D^*$ in $O(nm_{CD})$ time executing $n$ graph searches;
$D^{*}C$ in $O(nm_{CD})$ time using sparse matrix multiplication;
$B(D^{*}C)$ in $O(nn_k^{\omega-1})$ time using fast matrix multiplication;
$BD^{*}$ in $O(n_k^{\omega})$ time using fast matrix multiplication;
$(A-BD^{*}C)^{*}$ in $O(n_k^{\omega})$ time using fast matrix multiplication;
both matrix multiplications from~\eqref{eqn-tc} in $O(\omega(n,n_k,n))$ time, as $BD^{*}$ and $D^{*}C$
have one dimension of size $O(n_k)$.}
\fullversion{
In order to compute the transitive closure using this equation we compute:
\begin{itemize}
\item $D^*$ in $O(nm_{CD})$ time executing $n$ graph searches;
\item $D^{*}C$ in $O(nm_{CD})$ time using sparse matrix multiplication;
\item $B(D^{*}C)$ in $O(nn_k^{\omega-1})$ time using fast matrix multiplication;
\item $BD^{*}$ in $O(n_k^{\omega})$ time using fast matrix multiplication;
\item $(A-BD^{*}C)^{*}$ in $O(n_k^{\omega})$ time using fast matrix multiplication;
\item both matrix multiplications from~\eqref{eqn-tc} in $O(\omega(n,n_k,n))$ time, as $BD^{*}$ and $D^{*}C$
have one dimension of size $O(n_k)$.
\end{itemize}
} The theorem follows by plugging the bounds for $m_{CD}$ and $n_k$ to the list above.
\end{proof}

\fullversion{
\subsection{Determinants of Symmetric Matrices}
}
\exabstract{
\paragraph{Determinants and Linear Systems}
}

\fullversion{
Let us start from the Lanczos' algorithm, which is useful when dealing with sparse matrices.

\begin{theorem}[\cite{bini-pan,Golub:1996:MC:248979}]
\label{thm:lanczos}
There is a randomized algorithm,
which for a given generic symmetric square matrix $A$
in time $O(nm)$ computes $\det(A)$ together with matrices $Q, T$, such that
$A = Q T Q^T$,
where $T$ is a tridiagonal matrix, $n$ is the dimension of $A$
and $m$ is the number of non-zero entries in $A$.
\end{theorem}

Next, we show an algorithm computing a determinant of a matrix $M$,
corresponding to a given \dpl{} graph $G$. A symmetric matrix $M$ can
be seen as corresponding to the case when $G$ is undirected.
}

\exabstract{
We prove the following theorem regarding determinant of generic symmetric matrices.
The proof is similar to the proof of Theorem~\ref{thm-transitive-closure},
as we partition the matrix into blocks $A, B, C, D$, and use the formula for the determinant of a Schur complement
$\det(M) = \det(D) \det(A - B D^{-1} C).$
The determinant of the matrix $D$ and its decomposition is obtained by using
Lanczos' algorithm, and the running time is different than in Theorem~\ref{thm-transitive-closure}
as the matrix is symmetric, hence the block $B$ is as sparse as $C$.
}

\begin{theorem}
\label{theorem-determinant-symmetric}
Let $G$ be a directed \dpl{} graph with parameters $\pl < 2$ and $t$.
Let $M$ be a generic symmetric matrix, whose non-zero entries are a subset of non-zero entries of the adjacency matrix of $G$.
Then, we can compute the determinant of $M$ in $O((t+1)^{(\pl-1)+\frac{(\omega-2)(\pl-1)(2-\pl)}{(\omega-2)\pl+3-\omega}}n^{2+\frac{(\omega-2)(2-\pl)}{(\omega-2)\pl+3-\omega}})$ time.
\end{theorem}

\fullversion{
\begin{proof}
Similarly as in the proof for transitive closure, we start with sorting the rows of $M$
in decreasing order according to the number of non-zero entries
(which is upper bounded by the degree of a corresponding vertex of $G$).
Then, for a given $k \in [0,n]$, we split the matrix $M$ into $4$ submatrices

\begin{equation*}
M = \begin{bmatrix}
A & B \\
C & D
\end{bmatrix}
\end{equation*}

where $\begin{bmatrix} A & B \end{bmatrix}$ contains rows with more than $k$ non-zero entries,
$\begin{bmatrix} C & D \end{bmatrix}$ has at most $k$ non-zero entries in each row and $A$, $D$ are square matrices.
Let $m_{BCD}$ be the total number of non-zero entries in submatrices $B, C, D$.
As the matrix is symmetric $m_{BCD}$ is at most twice the number of non-zero entries in $\begin{bmatrix} C & D \end{bmatrix}$.
This in turn, by Lemma~\ref{lem:small_degrees}, is at most $O(n(t+1)^{\pl-1} k^{2-\pl})$.

By using the formula for the determinant of a Schur complement we obtain
$$\det(M) = \det(D) \det(A - B D^{-1} C)\,.$$
Let $n_k$ be the dimension of $A$ (which depends on $k$).
By Lemma~\ref{lem:high_degree_num} we infer
that $n_k$ is bounded by $O(n(t+1)^{\pl-1}k^{1-\pl})$,
which in turn gives
\begin{align}
\label{eq:2}
n_k = & O(n(t+1)^{\pl-1}k^{1-\pl})\,,
\end{align}
as $\pl \ge 1$.

By invoking Theorem~\ref{thm:lanczos} we compute $\det(D)$ as well
as matrices $Q, T$ such that $D = Q T Q^T$.
The running time needed is $O(n m_{BCD})$. 
Denote $Z = A - B D^{-1} C = A - B Q T Q^T C$.
To compute $Z$ efficiently, we first compute $Q^TC$ and $BQ$
in time $O(n m_{BCD})$ each, as both $C$ and $B$ are sparse,
i.e., have at most $m_{BCD}$ non-zero entries.
As $T$ is tridiagonal computing $T(Q^TC)$
takes time proportional to the size of $(Q^TC)$, that is $n \cdot n_k$.
Finally we multiply $BQ$ by $T(Q^TC)$ in time $O(n \cdot n_k^{\omega-1})$
by partitioning the matrices into $n / n_k$ submatrices of size $n_k \times n_k$ each
and invoking fast matrix multiplication on square matrices.
Finally, after computing $Z$, we can compute $\det(Z)$ in time $O(n_k^\omega)$.

Summing up, we have to set the value of $k$ to minimize the maximum of four values
\begin{itemize}
  \item $O(nm_{BCD})$ time used by invoking Theorem~\ref{thm:lanczos}
  and for computing the products $Q^TC$, $BQ$,
  \item $O(n \cdot n_k)$ time for computing the product $T \cdot (Q^TC)$,
  \item $O(n \cdot n_k^{\omega-1})$ time for computing the product $(BQ) \cdot (T(Q^TC))$,
  \item $O(n_k^{\omega})$ time for computing $\det(Z)$.
\end{itemize}
Note that $n \cdot n_k^{\omega-1}$ dominates both $n \cdot n_k$ and $n_k^{\omega}$.
Therefore, we need to set the value of $k$, so that $m_{BCD} = n_k^{\omega-1}$.
By using $(\ref{eq:1})$ and $(\ref{eq:2})$
we set $k = t^{\frac{(\omega-2)(\pl-1)}{(\omega-2)\pl+3-\omega}}n^{\frac{\omega-2}{(\omega-2)\pl+3-\omega}}$,
which finishes the proof of the theorem.
\end{proof}
}

\fullversion{
\subsection{Determinant of General Matrices}
\label{sec:general-det}
In the general case we use the following results due to
Eberly~\cite{eberly00}, who showed how the Frobenius normal form of a sparse matrix can be
computed. Frobenius normal form $F_A$ of a matrix $A$ is a block diagonal matrix with companion matrices
of monic polynomials $f_1, \ldots, f_k$ on the diagonal, where $f_i$ is divisible by $f_{i+1}$, for
$1 \le i \le k-1$ and $VAV^{-1} = F_A$. The companion matrix of a monic polynomial $x^d +
g_{d-1}x^{d-1} + \ldots + g_1 x + g_0 \in \mathcal{F}[x]$ is a $d \times d$ matrix defined as
\[
C_g = \left[\begin{matrix}
0 & \ldots & 0 & -g_0 \\
1 & \ldots & 0 & -g_1 \\
  &  \ddots & \vdots & \vdots \\
0 & \ldots & 1 & - g_{d-1}
\end{matrix}\right].
\]
The polynomials $f_1,\ldots, f_k$ are {\it the invariant factors} of $A$ and $k$ is the number of
invariant factors.

\begin{theorem}[\cite{eberly00}]
\label{theorem-eberly} There exists an algorithm for computing Frobenius normal form $F$ of the
matrix $A$ together with the transition matrix $V$ and its inverse with use of $O(n)$ matrix-vector
products and $\tilde{O}(kn^2)$ additional operations, where $k$ is the number of invariant
factors of $A$. The algorithm is randomized and may fail with arbitrarily small probability.
\end{theorem}

We will use the following preconditioning
due to~\cite{Kaltofen:1991} to make sure that there is just one invariant factor with high
probability. An $n\times n$ Hankel matrix $H$ is constructed from
$2n-1$ elements $h_0,\ldots,h_{2n-2}$ in the following way
\[
H = \left[\begin{matrix}
h_0& h_1&\ldots &h_{n-2} & h_{n-1}\\
h_1& h_2&\ldots &h_{n-1} & h_{n}\\
\vdots& & \iddots & & \vdots\\
h_{n-1}& h_n&\ldots &h_{2n-3} & h_{2n-2}\\
\end{matrix}
\right].
\]
We note that multiplication of a matrix by Hankel matrix $H$ takes
$\tilde{O}(n^2)$ time~\cite{bini-pan}. Similarly, computing $H^{-1}$ or multiplying
a matrix by $H^{-1}$ takes $\tilde{O}(n^2)$.

\begin{theorem}[Theorem 2 from \cite{Kaltofen:1991}]
\label{theorem-preconditioning}
Let $A$ be the non-singular square matrix $A$, let $H$ be a Hankel matrix with elements selected randomly and
uniformly from $\mathcal{F}$, then all leading (or trailing) principal submatrices of $\hat{A}=AH$ are non-singular with high probability.
\end{theorem}

\begin{theorem}[Equation (1) from \cite{Kaltofen:1991}]
\label{theorem-preconditioning-2}
Let $\hat{A}$ be matrix such that all its leading principal submatrices are non-singular,
let $J$ be a diagonal matrix with elements selected randomly and uniformly from $\mathcal{F}$, then $\tilde{A}=\hat{A}J$ has one invariant
factor with high probability.

\end{theorem}

}

\exabstract{
In the general case we use the result due to Eberly~\cite{eberly00},
who showed how the Frobenius normal form of a sparse matrix can be
computed. Together with other algebraic tools such as preconditioning from~\cite{Kaltofen:1991}
by using Hankel matrices we obtain the following theorem.
}

\begin{theorem}
\label{thm-general-determinant}
Let $G$ be a directed \dpl{} graph with parameters $\pl$ and $t$, and let $M$ be a matrix, whose non-zero entries are a subset of non-zero entries
of the adjacency matrix of $G$.
Moreover, let $1\le k< n^{1/\pl}(t+1)^{1-1/\pl}$.
Then, one can compute the determinant of $M$ in $O(n^2(t+1)^{\pl-1} k^{2-\pl}+\omega(n,n,n(t+1)^{\pl-1}k^{1-\pl}))$ randomized time.
\end{theorem}

\fullversion{
\begin{proof}
Similarly as in the symmetric case, we start with sorting the rows of $M$
in decreasing order according to the number of non-zero entries
Then, for a given $k \in [0,n]$, we split the matrix $M$ into $4$ submatrices

\begin{equation*}
M = \begin{bmatrix}
A & B \\
C & D
\end{bmatrix}
\end{equation*}

where $\begin{bmatrix} A & B \end{bmatrix}$ contains rows with more than $k$ non-zero entries, $\begin{bmatrix} C & D \end{bmatrix}$ has at most $k$ non-zero entries in each row and $A$, $D$ are square matrices.
Let $m_{CD}$ be the total number of non-zero entries in submatrices $C, D$. By Lemma~\ref{lem:small_degrees} $m_{CD}$ is bounded $O(n(t+1)^{\pl-1} k^{2-\pl})$.

Let $X$ be and arbitrary $n\times n$ matrix. A submatrices of $X$ obtained by performing similar split as for $M$ are denoted by
\begin{equation*}
X = \begin{bmatrix}
X_A & X_B \\
X_C & X_D
\end{bmatrix}.
\end{equation*}

Let $H$ and $J$ be random  matrices as given in Theorem~\ref{theorem-preconditioning}.
We cannot afford to precondition the whole matrix, so we  precondition
only the essential part that is needed for the Schur complement to work.

\begin{equation*}
\hat{M} =
\begin{bmatrix}
A & B \\
C & D
\end{bmatrix}
\begin{bmatrix}
I & H_B\\
0 & H_D
\end{bmatrix}
 =
\begin{bmatrix}
A & AH_B+BH_D \\
C & CH_B+DH_D
\end{bmatrix}.
\end{equation*}
You may observe that $AH_B+BH_D = (MH)_B$ and $CH_B+DH_D =(MH)_D$. Hence,
by Theorem~\ref{theorem-preconditioning} all trailing principal submatrices of $CH_B+DH_D$ are non-singular. Now
we apply the second part of the preconditioning.
\begin{equation*}
\tilde{M} =
\hat{M}
\begin{bmatrix}
I & 0 \\
0 & J_D
\end{bmatrix}
=
\begin{bmatrix}
A & (AH_B+BH_D)J_D \\
C & (CH_B+DH_D)J_D
\end{bmatrix}.
\end{equation*}
By Theorem~\ref{theorem-preconditioning-2} the matrix $\tilde{M}_D=(CH_B+DH_D)J_D$ has one invariant factor.
By using the formula for the determinant of a Schur complement we obtain
$$\det(\tilde{M}) = \det(\tilde{M}) \det(A - \tilde{M}_B (\tilde{M}_D)^{-1} C)\,.$$
Let $n_k$ be the dimension of $A$. By Lemma~\ref{lem:high_degree_num} we have that $n_k$
is bounded by $O(n(t+1)^{\pl-1}k^{1-\pl})$.

By invoking Theorem~\ref{theorem-eberly} we compute
matrices $V, F$ such that $\tilde{M}_D = V^{-1} F V$, as well as $\det(\tilde{M}_D)=\det(F)$ in time $\tilde{O}(n m_{CD})$.

Denote $\tilde{Z} = A - \tilde{M}_B (\tilde{M}_D)^{-1} C= A -  \tilde{M}_B V^{-1}F^{-1} V C$.
To compute $\tilde{Z}$ efficiently, we first compute $\tilde{M}_B=(AH_B+BH_D)J_D=(A,B)(H_B,H_D)^TJ_D$ what
requires $O(n^2)$ time. Then computing $\tilde{M}_B V^{-1}$ requires $O(\omega(n,n_k,n))$ time.
Next, we compute $VC$ in $O(n m_{CD})$ time because $C$ has at most $m_{CD}$ non-zero entries.

Due to special structure of $F$ we know that $F^{-1}$ has $O(n)$ non-zero entries and
computing it takes $O(n)$ time using the following form for each companion matrix
\[
C_g^{-1} =\frac{1}{g_0} \left[\begin{matrix}
-g_1 & 1 & 0 &\ldots & 0 \\
-g_2 & 0 & 1 & \ldots & 0 \\
\vdots  &  \vdots & \ddots &  \vdots \\
-g_{d-1} & 0 & 0& \ldots & 1\\
-1 &0 &0&  \ldots & 0
\end{matrix}\right].
\]

As $F^{-1}$ has $O(n)$ non-zero entries so computing $F^{-1} V C$
takes time proportional to the size of $VC$, that is $n \cdot n_k$.
To obtain $\tilde{Z}$ we multiply $\tilde{M}_B V^{-1}$ by $F^{-1}VC$ in time $O(n / n_k \cdot n_k^{\omega})$.
After computing $\tilde{Z}$, we can compute $\det(\tilde{Z})$ in time $O(n_k^{\omega})$.
Finally, by our preconditioning $\det(M)=\frac{\det(\tilde{Z})\det(D)}{\det(H_DJ_D)}$, where we
need $\tilde{O}(n)$ time to compute $\det(H_DJ_D)$.

Summing up, we have to set the value of $\param$ to minimize the maximum of the following
\begin{itemize}
\item $\tilde{O}(n^2)$ for computing $BHJ$, $DHJ$ and $HJ V^{-1}$,
\item $\tilde{O}(nm_{CD})$ time used by invoking Theorem~\ref{theorem-eberly}
and for computing the product $VC$,
\item $O(\omega(n,n,n_k))$ time for computing $B(HJ V^{-1})$,
  \item $O(nn_k)$ time for computing the product $F^{-1}\cdot VC$,
  \item $O(nn_k^{\omega-1})$ time for computing the product $BHJ V^{-1} \cdot F^{-1}VC$,
  \item $O(n_k^{\omega})$ time for computing $\det(\tilde{Z})$,
  \item $\tilde{O}(n)$ time for computing $\det(HJ)$ and $F^{-1}$.
\end{itemize}
Note that $O(\omega(n,n,n_k))$ dominates $O(nn_k^{\omega-1})$ , $nn_k$ and $n_k^{\omega}$,
whereas $\tilde{O}(nm_{CD})$ dominates $\tilde{O}(n^2)$. Therefore we need to set the value of $\param$, so
that $nm_{CD} = \omega(n,n,n_k)$.
\end{proof}
}

The above theorem gives a general statement that in the parameter range $1<\pl<2$ it is possible to compute determinant
faster than by using algorithms for general graphs. However, the statement of the theorem contains tangled equation, so in order to simplify it
we assume that $t=O(\polylog{} n)$.  Let $\omega(\beta)$ be defined such that $n^{\omega(\beta)}=\omega(n,n,n^\beta)$.

\begin{corollary}
\label{cor-general-determinant}
Let $G$ be a directed \dpl{} graph with parameters $\pl$ and $t$, and $M$ be a matrix, whose non-zero entries are a subset of non-zero entries of the adjacency matrix of $G$.
Let $0\le \beta< 1$ be such that $2+\beta(2-\pl)= \omega(1+\beta(1-\pl))$, $\beta<1/\pl$. Moreover, assume $t=O(\polylog{} n)$. Then,
we can compute the determinant of $M$ in $\tilde{O}(n^{2+\beta(2-\pl)})$ randomized time.
\end{corollary}

\fullversion{
We observe that when $M$ is symmetric then $B$ is sparse as in the proof of Theorem~\ref{theorem-determinant-symmetric}. In
such a case computing $B(HJ V^{-1})$ takes $O(nm)$ time instead of $\omega(n,n,n_k)$ time and we obtain similar bounds as in Theorem~\ref{theorem-determinant-symmetric}.
}

\exabstract{
By inspecting the proof of Theorem~\ref{theorem-determinant-symmetric}
we obtain the following corollaries.
}
\begin{corollary}
\label{cor-symmetric-determinant}
It is possible to drop the generic assumption from Theorem~\ref{theorem-determinant-symmetric} by
increasing the running time by polylogarithmic factors.
\end{corollary}

\fullversion{
\subsection{Linear System Solution and Matrix Inverse}
In order to solve linear system with matrix $M$ we will extend the idea of the algorithm from the previous section,
i.e., we first run the above algorithm to compute the determinant of $M$ and store all intermediate results of this
computation. Let $v$ be an $n$ length vector, then to find a vector $x$ such that $Mx=v$ we compute
\begin{equation}
x=M^{-1}v =
\begin{bmatrix}
I & H_B\\
0 & H_D
\end{bmatrix}
\begin{bmatrix}
I & 0 \\
0 & J_D
\end{bmatrix}
\tilde{M}^{-1}
v.
\label{eqn-linear-system}
\end{equation}
Now we express inverse of a $\tilde{M}$  in the block form
\begin{equation}
\label{equ:block_inverse}
\tilde{M}^{-1}=
\begin{bmatrix}
I & 0 \\
- (\tilde{M}_D)^{-1}C & I
\end{bmatrix}
\begin{bmatrix}
\tilde{Z}^{-1} & 0 \\
C & (\tilde{M}_D)^{-1}
\end{bmatrix}
\begin{bmatrix}
I & - \tilde{M}_B(\tilde{M}_D)^{-1} \\
0 & I
\end{bmatrix}.
\end{equation}
Now we plug in the equation $(\tilde{M}_D)^{-1} = V^{-1}F^{-1}V$ to obtain
\begin{equation*}
\tilde{M}^{-1}v=
\begin{bmatrix}
I & 0 \\
- V^{-1}F^{-1}VC & I
\end{bmatrix}
\begin{bmatrix}
\tilde{Z}^{-1} & 0 \\
C & V^{-1}F^{-1}V
\end{bmatrix}
\begin{bmatrix}
I & - \tilde{M}_BV^{-1}F^{-1}V \\
0 & I
\end{bmatrix}v.
\end{equation*}
Observe that all matrices in the above have been computed during the computation of the determinant, so
computing $\tilde{M}^{-1}v$ takes $O(n^2)$ time. Then we compute $M^{-1}v$ using \eqref{eqn-linear-system} in $\tilde{O}(n^2)$ time.
}

\begin{corollary}
Theorem~\ref{theorem-determinant-symmetric}, Theorem~\ref{thm-general-determinant} and Corollary~\ref{cor-symmetric-determinant}
can be extended to compute a solution to linear system at the cost of using $\tilde{O}(n^2)$ additional time.
\end{corollary}

\fullversion{
Finally, we observe that using \eqref{equ:block_inverse} we can compute the inverse matrix explicitly. In the case
of generic matrices it takes the same time as needed for transitive closure using \eqref{eqn-tc}, whereas
in the symmetric case the most expensive multiplication takes $O(nm)$ time instead of $O(\omega(n,n,n_k))$ time, so we
obtain the following corollary.}
\exabstract{
Finally, we observe that using block matrix decomposition, similar to the one used for transitive closure,
we can compute the inverse matrix explicitly, so we
obtain the following corollary.
}

\begin{corollary}
Theorem~\ref{theorem-determinant-symmetric}, Theorem~\ref{thm-general-determinant} and Corollary~\ref{cor-symmetric-determinant}
can be extended to compute a inverse matrix in the same asymptotic time.
\end{corollary}

\fullversion{\subsection{Perfect Matching}}
\exabstract{\paragraph{Perfect Matching}}
As observed by Lovasz~\cite{Lovasz79} in 1979 checking whether a graph contains a perfect matching can be done in $O(n^{\omega})$ time
using one determinant computation for an appropriately defined skew-symmetric matrix. However, an algorithm for finding such perfect matching was shown 25 years latter~\cite{ms04}.
Here, we reuse this idea in the case of \dpl{} graphs to check whether a graph contains a perfect matching and to find one. The running
time of the resulting algorithms is the same as the running time of determinant computation for symmetric matrices.

\fullversion{
Let us define for a graph $G$ a skew-symmetric adjacency matrix $\tilde{M}$ in the following way
\[
\tilde{M}_{i,j} =
\left\{
  \begin{array}{ll}
    z_{i,j}  & \textrm{if } ij \in E \textrm{ and } i < j \\
    -z_{j,i} & \textrm{if } ij \in E \textrm{ and } i > j \\
    0 & \textrm{ otherwise}
  \end{array}
\right. ,
\]
where for each edge $ij\in E$ the variables $z_{i,j}$ are distinct.

\begin{theorem}[\cite{Lovasz79}]
\label{theorem-tutte}
Let $M$ be a matrix obtained from $M$ by substituting uniformly at random elements from $\mathcal{F}$ for variables.
If $G$ has a perfect matching then $\det(M)\neq 0$ with high probability, whereas when $G$ has no perfect matching then $\det(M)=0$.
\end{theorem}

We can observe that in our derivation of Corollary~\ref{cor-symmetric-determinant} we have used only the fact
that non-zero structure of the matrix is symmetric, so the same time bounds hold for computation of a determinant of a skew-symmetric matrix $M$.
Hence, assuming that $t=O(\polylog{} n)$, to test whether a \dpl graph contains a perfect matching we need
$O(n^{2+\frac{(\omega-2)(2-\pl)}{(\omega-2)\pl+3-\omega}})$ randomized time. This is faster then Micali-Vazirani algorithm~\cite{mv}, that
works in $O(n^{1/2 + 2/\pl})$ time, when $\pl\le \frac{\sqrt{-7 w^2+26 w-15}+3 w-9}{2w-4} < 1.09042$.
However, what is left is to construct the perfect matching when we know that the graph contains one. In order to do it
we need to compute $(M^{-1})_A = Z^{-1} = (A - B D^{-1} C)^{-1}$. By our preconditioning we have
\begin{equation*}
\label{eqn-linear-system-2}
M^{-1} =
\begin{bmatrix}
I & H_B\\
0 & H_D
\end{bmatrix}
\begin{bmatrix}
I & 0 \\
0 & J_D
\end{bmatrix}
\tilde{M}^{-1}.
\end{equation*}
Using this equation and the equality $(\tilde{M}^{-1})_A = \tilde{Z}^{-1}$ we obtain
\[
M^{-1}_A = Z^{-1} = \tilde{Z}^{-1} - H_B J_DV^{-1}F^{-1}VC.
\]
We need $O(nm_{BCD})$ time to compute $VC$, $O(nn_k)$ time for $(H_B(J_DV^{-1}))F^{-1}$ and
finally $O(n \cdot n_k^{\omega-1})$ time to compute $(H_B J_DV^{-1}F^{-1})(VC)$. These running
times do not increase the running times stated in Theorem~\ref{theorem-determinant-symmetric} and Corollary~\ref{cor-symmetric-determinant}.
Let $G_A$ be the subgraph of $G$ represented by $\tilde{M}_A$, i.e., the subgraph induced only
 by vertices of degree at least $k$. Now we can apply the following observation that is the core idea of~\cite{Harvey06}.

\begin{theorem}[Procedure DeleteEdgesWithin from~\cite{Harvey06}]
Given a Schur complement of $M_A$ (i.e., $Z^{-1}$) one can find in $O(n_k^{\omega})$ time a set of edges $P_A\subseteq E(G_A)$
such that $P_A$ can be extended to a perfect matching $P$ of the whole graph $G$.
\end{theorem}

Hence, we first invoke the above theorem to find $P_A$ and then we are left to find matching in a graph $G-E(G_A)$, where
all edges between high degree vertices have been removed. This way the degree of a vertex in $G-E(G_A)$ is bounded by $k$ and
so the total number of edges is $O(m_{BCD})$. Using the algorithm by Micali and Vazirani we need $O(\sqrt{n}m_{BCD})=O(nm_{BCD})$ time to extend $P_A$ to a
perfect matching on $G$.  This way we obtain the following theorem.
}

\begin{theorem}
\label{theorem-matchings}
Let $G$ be a \dpl{} graph with parameters $\pl < 2$ and $t$.
Then, we can find a perfect matching in $G$ in $O((t+1)^{(\pl-1)+\frac{(\omega-2)(\pl-1)(2-\pl)}{(\omega-2)\pl+3-\omega}}n^{2+\frac{(\omega-2)(2-\pl)}{(\omega-2)\pl+3-\omega}})$ time with high probability.
\end{theorem}

\fullversion{\subsection{Complexity of PageRank}}
\exabstract{\paragraph{Complexity of PageRank}}
Let us now discuss the arithmetic complexity of exact PageRank computation.\footnote{We note that the approximate iterative methods used in practice
have worse theoretical running time bounds.} Computing PageRank is
a simple version of the eigenproblem, where we are asked to find eigenvector for the eigenvalue which
is equal to $1$. Eigenproblems, in comparison with the determinant problem, is usually more
challenging, because we cannot easily  use preconditioning, as it can change both eigenvalues and eigenvectors.
The complexity of this simple problem is either $O(n^{\omega})$ using~\cite{959911} or $O(k nm)$ using~\cite{eberly00} (where $k$ is number
of invariant factors). However, we can show that the problem is easier on a directed \dpl graph.

\fullversion{
We assume that we are given a graph $G$ where out degrees satisfy power law with exponent $\pl$. The PageRank vector $\pi$ is
the eigenvector of the following $n \times n$ matrix
\begin{equation}
M = cP + (1-c) 1/n E
\label{eq:prdef}
\end{equation}
where $c$ is a damping factor which can be set between $0$ and $1$ (typically $0.85$), matrix $P$ is an adjacency matrix of $G$ defined with rules:
$P_{ij} = 0$ if there are no edges from $i$ to $j$ and $P_{ij} = 1/\outdeg(j)$ otherwise and $E$ is a matrix whose all entries equal one.
In other words, having defined $M$, the PageRank vector is an eigenvector $\pi$ of $M$ corresponding to eigenvalue $\mu = 1$, i.e.,
the PageRank vector satisfies the following equations
\begin{align}
\pi^T &= \pi^T M \\
\pi^T e &= 1
\label{eq:pagerank}
\end{align}
where $e$ is a vector of size $n$ whose all entries equal one.

As usual we start with sorting the rows of $P$ in decreasing order according to the number of non-zero entries, i.e., out-degrees of
corresponding vertices. Let us consider the following matrix decomposition
\begin{equation*}
P = \begin{bmatrix}
A & B \\
C & D
\end{bmatrix},
\end{equation*}
where $[A,B]$ contains rows with more than $k$ non-zero entries for $1\le k \le n$. As previously the size of $A$ is denoted by $n_k$, whereas
the number of non-zero entries in $D$ is denoted by $m_D$.
We define as well the corresponding decomposition of $\pi$
\begin{equation*}
 \pi=(\pi_A, \pi_D)
\end{equation*}

Using this block form we can rewrite the equation~\ref{eq:prdef} as follows
\begin{equation*}
M = \begin{bmatrix}
cA + (1-c)/n E_A & cB + (1-c)/n E_B \\
cC + (1-c)/n E_C & cD + (1-c)/n E_D
\end{bmatrix}
\end{equation*}
where $E_A$, $E_B$, $E_C$ and $E_D$ are matrices of the appropriate size whose all entries equal one.
In turn, the equation~\eqref{eq:pagerank} can be rewritten as
\begin{equation*}
\pi^T (I-M)=
\pi^T\begin{bmatrix}
I-M_A & M_B \\
M_C & I-M_D
\end{bmatrix} = 0.
\end{equation*}
We observe that $I-M$ is an irreducible $M$-matrix.\footnote{A matrix $X$ is said to be {\em $M$-matrix} when it can be written as $X=aI-Y$
where all entries of $Y$ are nonnegative and $a$ is greater or equal then the spectral radius of $B$. The spectral radius of
$M$ is $1$ as it is a stochastic matrix.} This implies by Theorem 4.16 from~\cite{berman} that every leading or trailing principal and proper
submatrix of $I-M$ is nonsingular. In particular $I-M_D$ is nonsingular. Let $H_D$ be a random Hankel matrix and let $J_D$ be a random
diagonal matrix. Then by the preconditioning from Section~\ref{sec:general-det} $(I-M_D)H_DJ_D$ has one invariant factor.
We observe that multiplication of the matrix
$I-M_D = I-cD + (1-c)/n E_D$  by a vector takes $O(m_D)$ arithmetic operations, so by Theorem~\ref{theorem-eberly} we can compute Frobenius normal $F$ and
the transition matrix $V$ of $(I-M_D)H_DJ_D$ in $\tilde{O}(nm_D)$ arithmetic operations.

The stochastic complement of $M_A$ in $I-M$ is the following matrix
\begin{equation*}
\label{eqn-lazy-form}
 S_{A} = M_{A} + M_{B} (I- M_D)^{-1} M_{C} = M_{A} + M_{B} J_DH_D V^{-1}F^{-1} V M_{C}. 
\end{equation*}
Computing $S_A$ requires $O(nm_D)$ arithmetic operations to compute $VM_{C}$, $\tilde{O}(nn_k)$ arithmetic operations to compute $F^{-1}(V M_C)$ and $M_BJ_DH_D$,
$O(\omega(n,n,n_k))$ arithmetic operations to compute $(M_B J_DH_D)V^{-1}$, and finally $O(n n_k^{\omega-})$ arithmetic operations to compute
$(M_{B} J_DH_D V^{-1})(F^{-1} V M_{C})$.

Now using the equations from~\cite{Meyer89stochasticcomplementation} we obtain
\begin{equation*}
\pi_A^T S_{A} = \pi_A^T,
\end{equation*}
\begin{equation*}
\pi_D^T = \pi_A^T M_B (I- M_D)^{-1} = \pi_A^T M_B J_DH_D V^{-1}F^{-1} V,
\end{equation*}
which means that $\pi_A$ is a stationary distribution for the smaller matrix $S_{A}$ and can be computed in
$O(n_k^{\omega})$ arithmetic operations~\cite{KellerGehrig1985309}. Then in order to compute $\pi_D$ we need $O(n^2)$ arithmetic operations.
We note that $O(\omega(n,n,n_k))$ dominates $O(nn_k^{\omega-1})$ , $O(nn_k)$ and $O(n_k^{\omega})$,
whereas $\tilde{O}(nm_D)$ dominates $\tilde{O}(n^2)$. Therefore we obtain the following theorem.
}

\begin{theorem}
\label{thm-page-rank}
Let $G$ be a directed \dpl{} graph with parameters $\pl$ and $t$, and let $1\le k< n^{1/\pl}(t+1)^{1-1/\pl}$. Then, we can compute PageRank of $G$ with $O(n^2(t+1)^{\pl-1} k^{2-\pl}+\omega(n,n,n(t+1)^{\pl-1}k^{1-\pl}))$ arithmetic operations with high probability.
\end{theorem}

We note that the above idea can be combined together with iterative methods. In such a case instead of using the stochastic complement of
$M_A$ we shall use the stochastic complement of $M_D$, i.e., $S_D= M_{D} + M_{C} (I- M_A)^{-1} M_{B}$. However, we explicitly compute only
the inverse $(I-M_A)^{-1}$ using Strassen's fast matrix inverse, but we not execute
other multiplications and keep $S_D$ in the lazy form as given by this equation. We can apply iterative
methods to compute the stationary distribution of $S_D$ using this lazy form. We have implemented this approach and on a single computer it can reduce
the time needed for PageRank computation by a factor of two on graphs that have approximately $100 000$ nodes, e.g., WikiTalk network.

\bibliographystyle{abbrv}
\bibliography{references}

\begin{table}
\centering
\scriptsize
\begin{tabular}{@{}l|c|c|c|c|c|c|c@{}}
Graph & $n$ & $m$ & $c_1$ & $c_2$ & $\alpha$ & $t$ & $\Delta / \sqrt{m}$\\
\hline
\hline
Amazon (directed, in-degree + out-degree) 	& $241761$ & $1131217$ &	$5$	& $0.615102$ &	$3.198$	&	$22.2994$ & $0.3996$	\\
\hline
AstroPh (directed, in-degree + out-degree) 	& $17903$ & $393944$ &	$1.2888$	& $0.208271$ &	$2.0189$	&	$21.0207$ & $1.606$	\\
\hline
Cities (directed, in-degree + out-degree)	& $3144$ & $34753$ & 	$0.9652$	& $-$		& $1.9126$	& $1.8661$ & $10.6425$\\
\hline
CondMatt (undirected) 	& $21363$ & $182572$ &	$4.7952$	& $2.15535$ &	$5.2849$	&	$26.1942$ & $0.65296$	\\
\hline
Dblp (undirected) 	& $718115$ & $5573812$ &	$2.6633$	& $5.80862$ &	$3.4134$	&	$9.589$ & $0.3808$	\\
\hline
Enron (undirected) 	& $33696$ & $361622$ &	$1.2549$	& $0.610801$ &	$2.2674$	&	$3.4682$ & $2.2998$	\\
\hline
Epinions (directed, in-degree + out-degree) 	& $32223$ &  $443506$ &	$1.2166$	& $-$ &	$1.8863$	&	$3.8008$ & $4.1894$	\\
\hline
EuAll (directed, in-degree + out-degree) 	& $34203$ & $151132$ &	$3.1966$	& $1.37106$ &	$2.4201$	&	$3.756$ & $3.8635$	\\
\hline
Facebook (undirected) 	& $59691$ & $1456818$ &	$0.8077$	& $-$ &	$1.6668$	&	$6.2728$ & $0.8409$	\\
\hline
HepPh (directed, in-degree + out-degree) 	& $12711$ & $139965$ &	$5$	& $3.00723$ &	$5.2231$	&	$70.1391$ & $1.0104$	\\
\hline
LiveJournal (directed, in-degree + out-degree)  & $3828682$ & $65349587$ &	$2.1985$	& $4.64595$ &	$2.5893$	&	$18.8438$ & $2.8287$	\\
\hline
NotreDame (directed, in-degree + out-degree) 	& $53968$ & $296228$ &	$2.2113$	& $1.55224$ &	$2.6274$	&	$9.5051$ & $14.0243$	\\
\hline
Slashdot (directed, in-degree + out-degree) 	& $71307$ & $841201$ &	$1.4678$	& $0.166008$ &	$2.0236$	&	$3.8451$ & $5.5191$	\\
\hline
WikiTalk (directed, in-degree + out-degree) 	& $111881$ & $1477893$ &	$1.5124$	& $0.177209$ &	$2.031$	&	$3.7847$ & $6.6613$	\\
\hline
WIW (undirected)	& $29406$ & $393797$ &	$0.5474$	&	$-$		& $1.2562$	& $0$ & $1.0135$	\\
\hline
YouTube (undirected) 	& $495957$ & $3873496$ &	$1.0395$	& $0.66258$ &	$2.2474$	&	$1.8672$ & $12.9103$	\\
\hline
\hline						
AstroPh (directed, out-degree) 	& $17903$ & $393944$ &	$3.1737$	& $ 4.6716$ &	$3.5199$	&	$32.734$ &  $-$	\\
\hline
Epinions (directed, out-degree) 	& $32223$ &  $443506$ &	$2.0569$	& $1.21022$ &	$2.4379$	&	$6.3751$ &  $-$	\\
\hline
EuAll (directed, out-degree) 	& $34203$ & $151132$ &	$2.4122$	& $0.401895$ &	$2.1407$	&	$0$ &  $-$	\\
\hline
HepPh (directed, out-degree) 	& $12711$ & $139965$ &	$4.3101$	& $1.10021$ &	$4.7202$	&	$25.2953$ &  $-$	\\
\hline
LiveJournal (directed, out-degree) 	& $3828682$ & $65349587$ &	$2.2261$	& $8.05048$ &	$2.7745$	&	$12.0126$ &  $-$	\\
\hline
NotreDame (directed, out-degree) 	& $53968$ & $296228$ &	$4.9269$	& $1.65396$ &	$2.6162$	&	$0.5484$ &  $-$	\\
\hline
Slashdot (directed, out-degree) 	& $71307$ & $841201$ &	$1.5542$	& $0.376638$ &	$2.165$	&	$3.3024$ &  $-$	\\
\hline
WikiTalk (directed, out-degree) 	& $111881$ & $1477893$ &	$1.1869$	& $-$ &	$1.9364$	&	$0.9833$ &  $-$	\\
\hline
\hline									
Amazon (directed, in-degree + out-degree)	& $241761$ & $1131217$ &	$5$	& $-$ &	$1.8072$	&	$0$ &  $0.3996$	\\
\hline
CondMatt (undirected)	& $21363$ & $182572$ &	$5$	& $0.420346$ &	$2.1699$	&	$0$ & $0.65296$	\\
\hline
\end{tabular}
\caption{Adjustment of \dpl{} universal constants.}
\label{table:plb_adjustment}
\end{table}
\begin{figure*}
\centering
\footnotesize
\subfigure[\bf{CondMatt (undirected)}]{\includegraphics[scale=0.32]{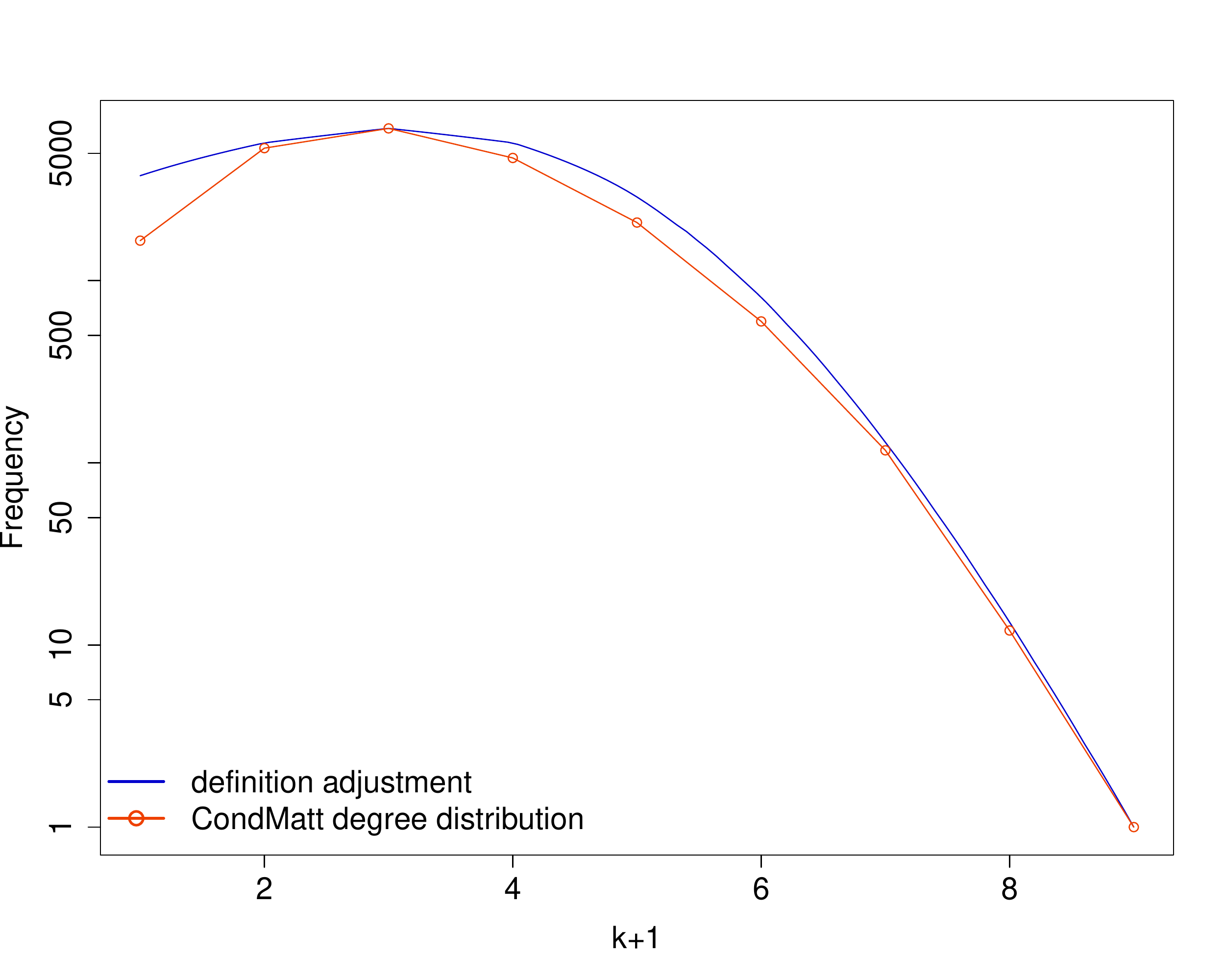}}
\subfigure[\bf{Epinions (directed, in-degree + out-degree)}]{\includegraphics[scale=0.32]{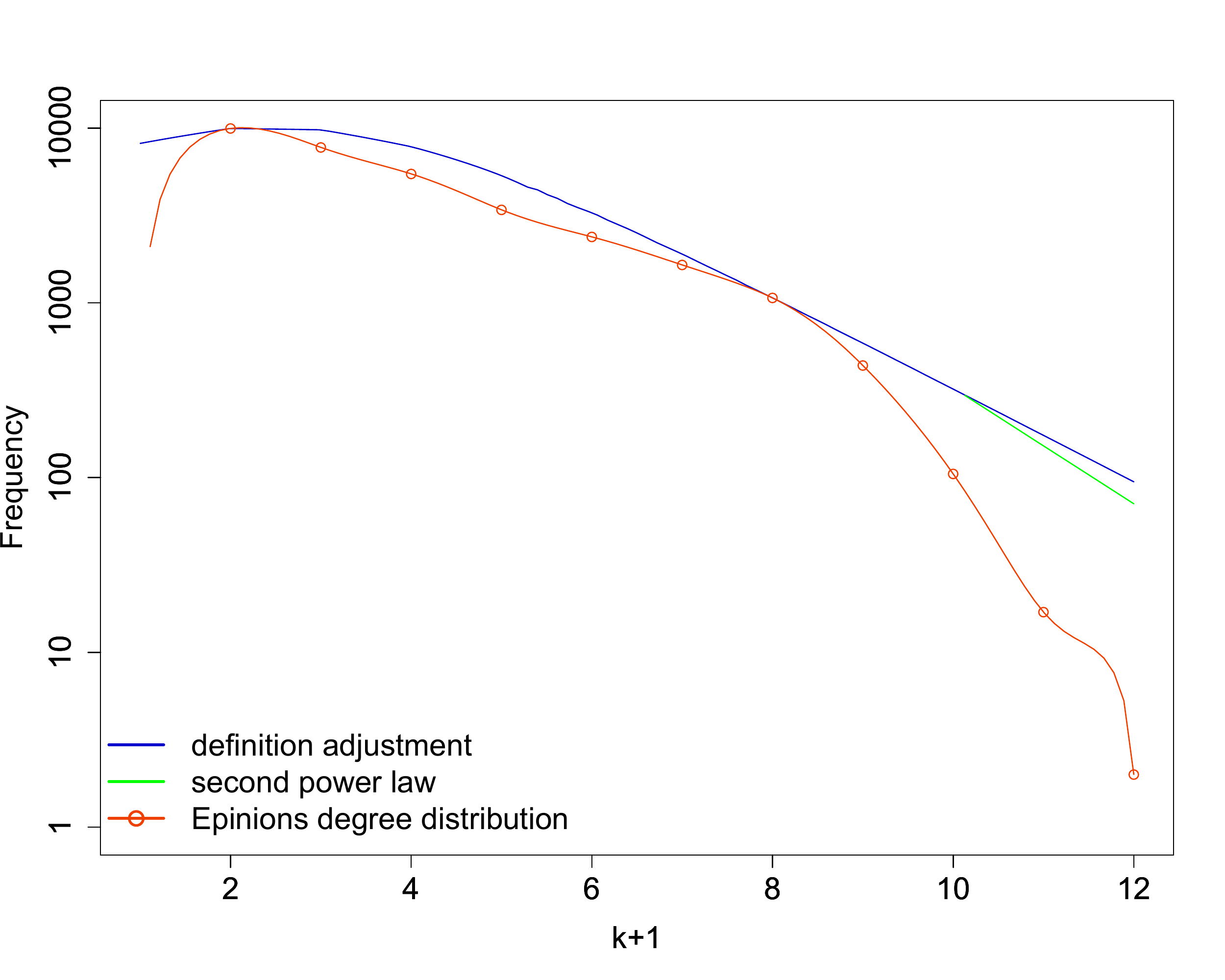}}
\subfigure[\bf{Facebook (undirected)}]{\includegraphics[scale=0.32]{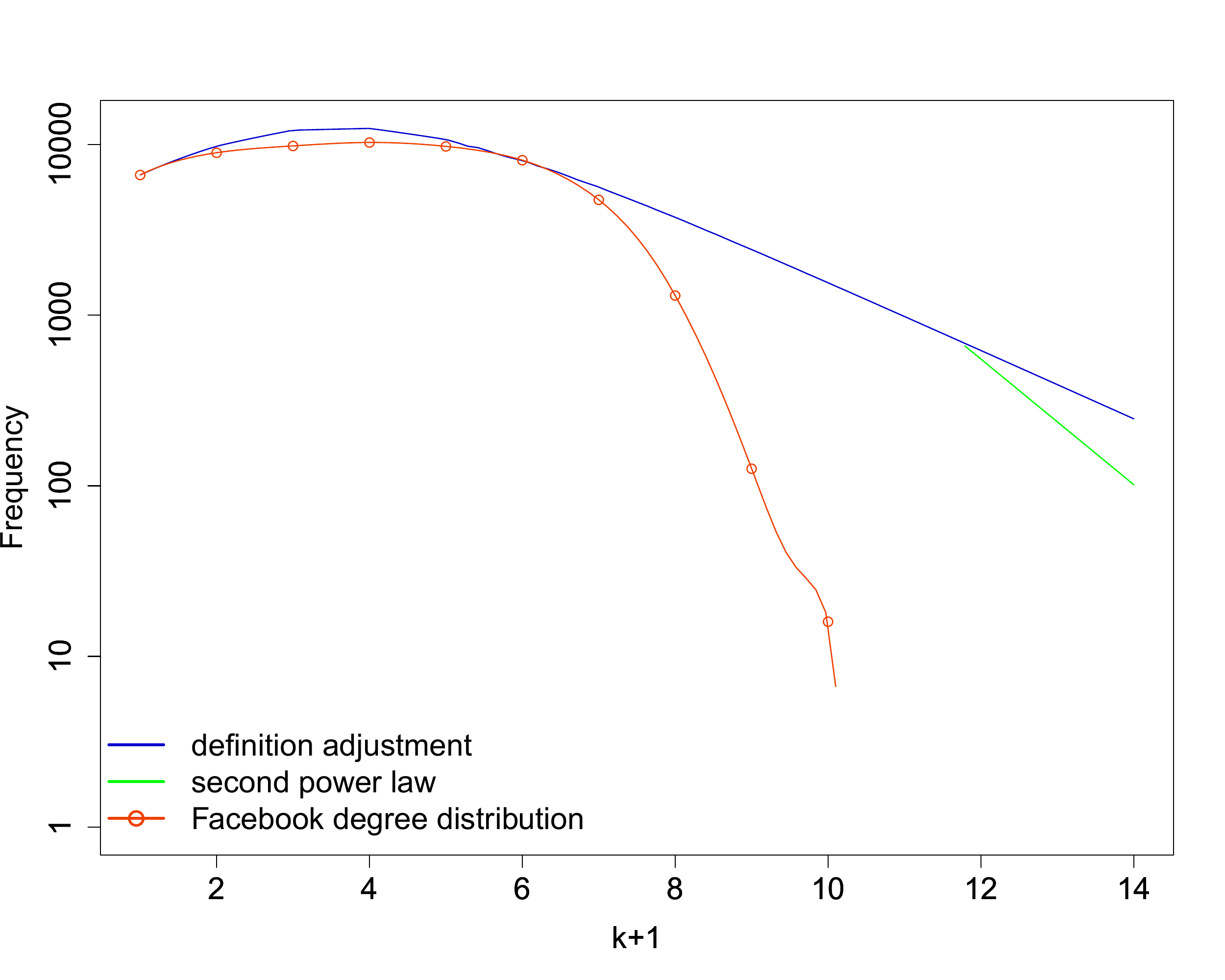}}
\subfigure[\bf{WikiTalk (directed, out-degree)}]{\includegraphics[scale=0.32]{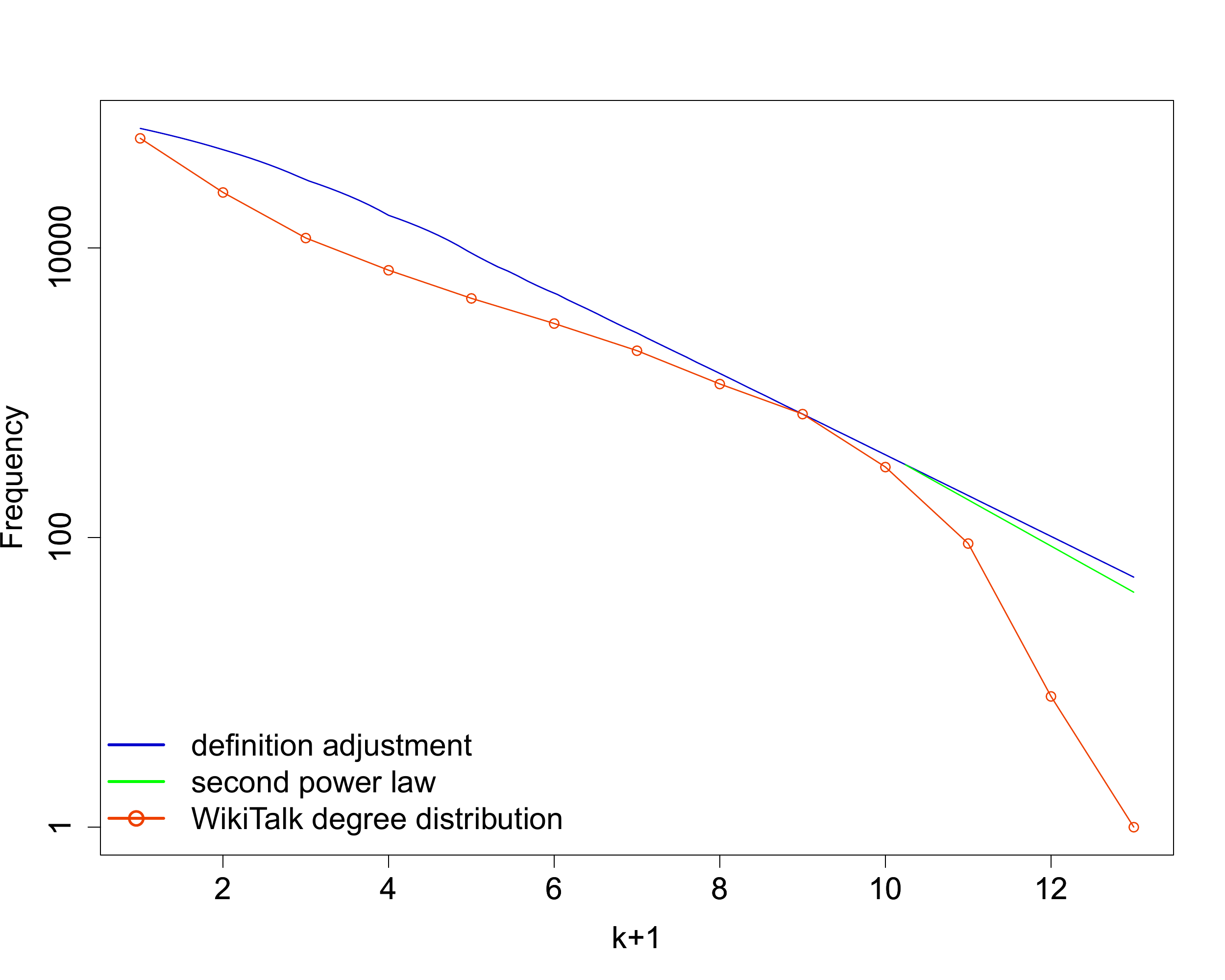}}
\subfigure[\bf{NotreDame (directed, out-degree)}]{\includegraphics[scale=0.32]{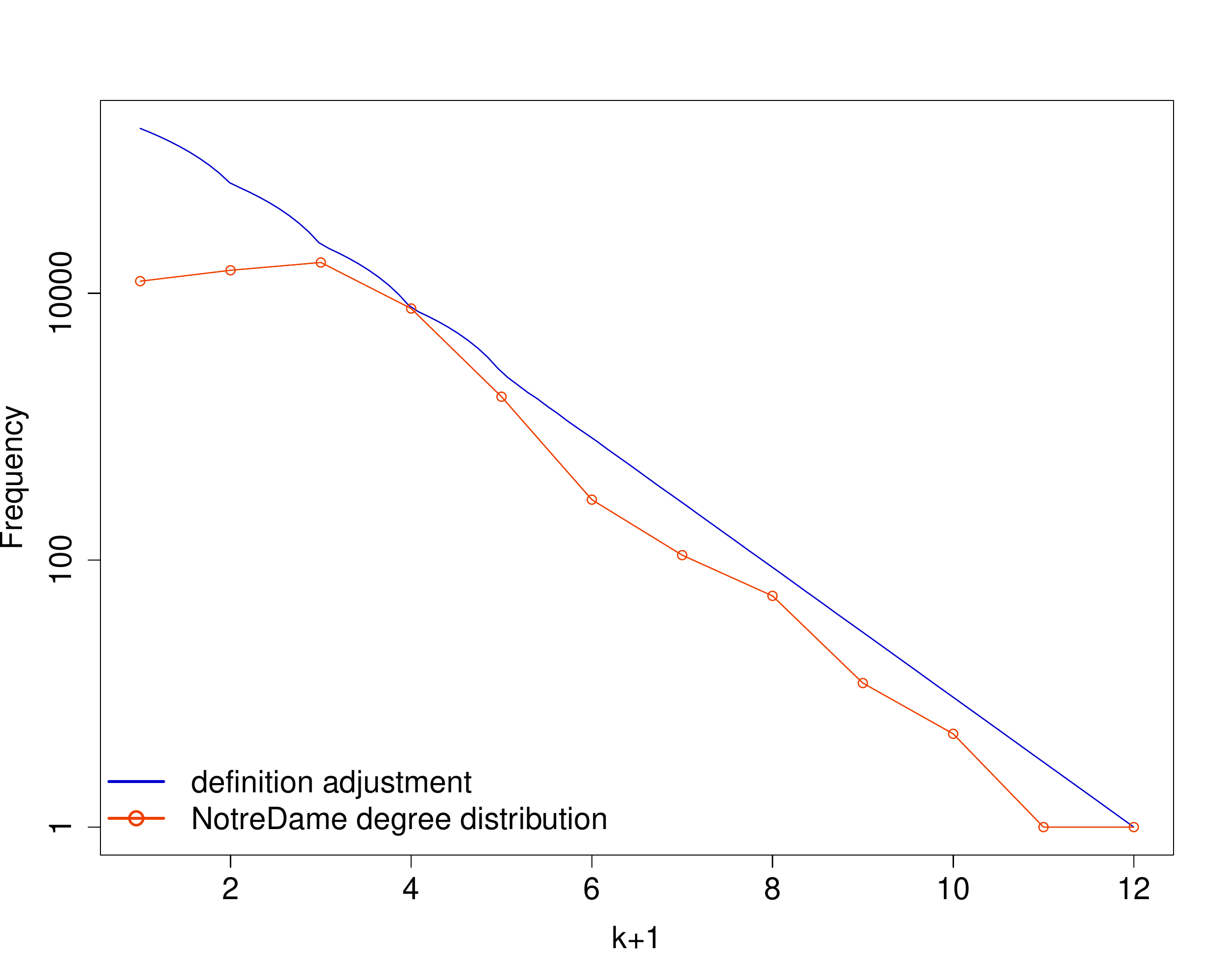}}
\subfigure[\bf{LiveJournal (directed, out-degree)}]{\includegraphics[scale=0.32]{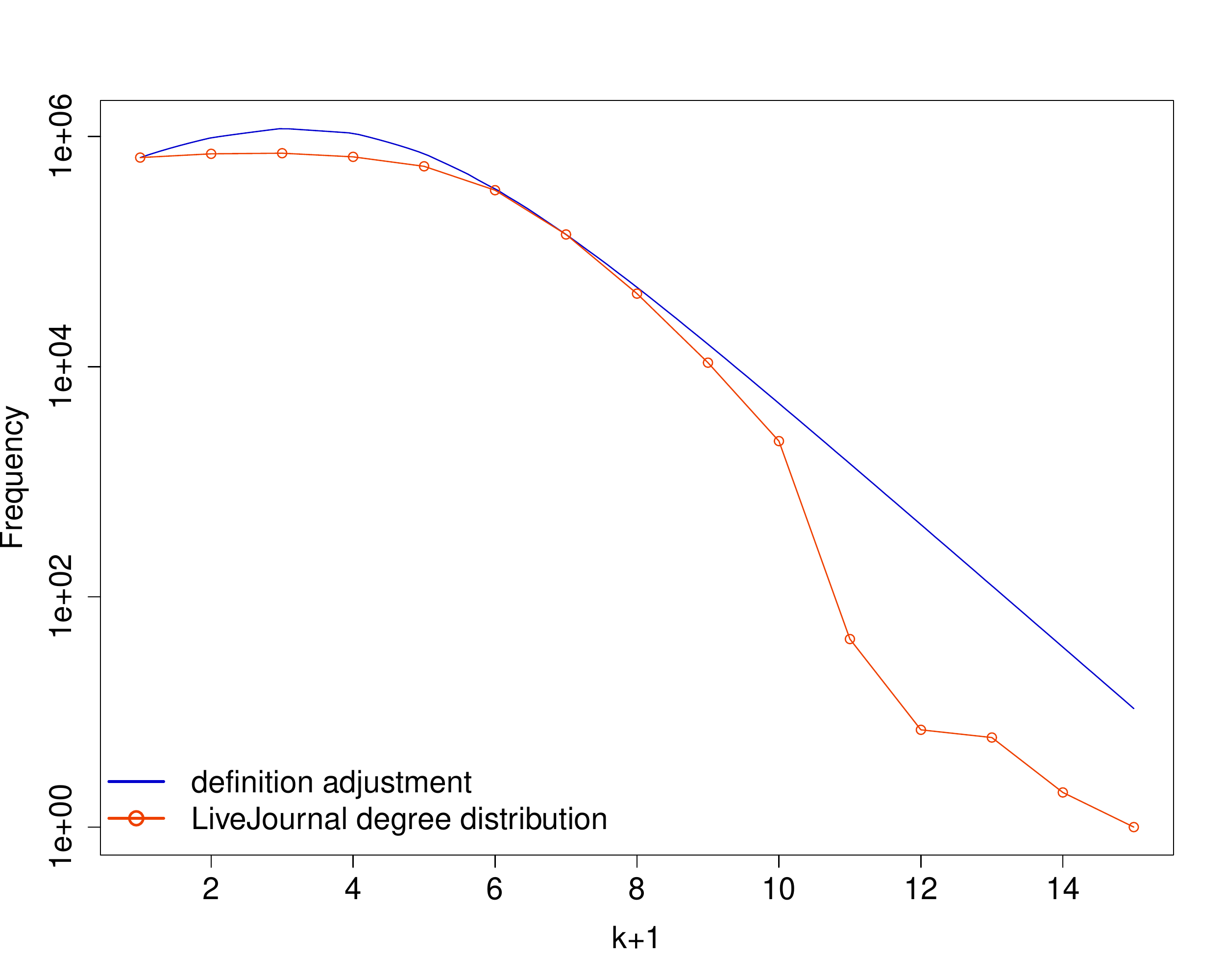}}
\caption{Real-World networks are \dpl{}: definition adjustment}
\label{fig:dpl_graphs}
\end{figure*}

\begin{figure*}
\centering
\footnotesize
\subfigure[\bf{Amazon (directed, in-degree + out-degree)}]{\includegraphics[scale=0.32]{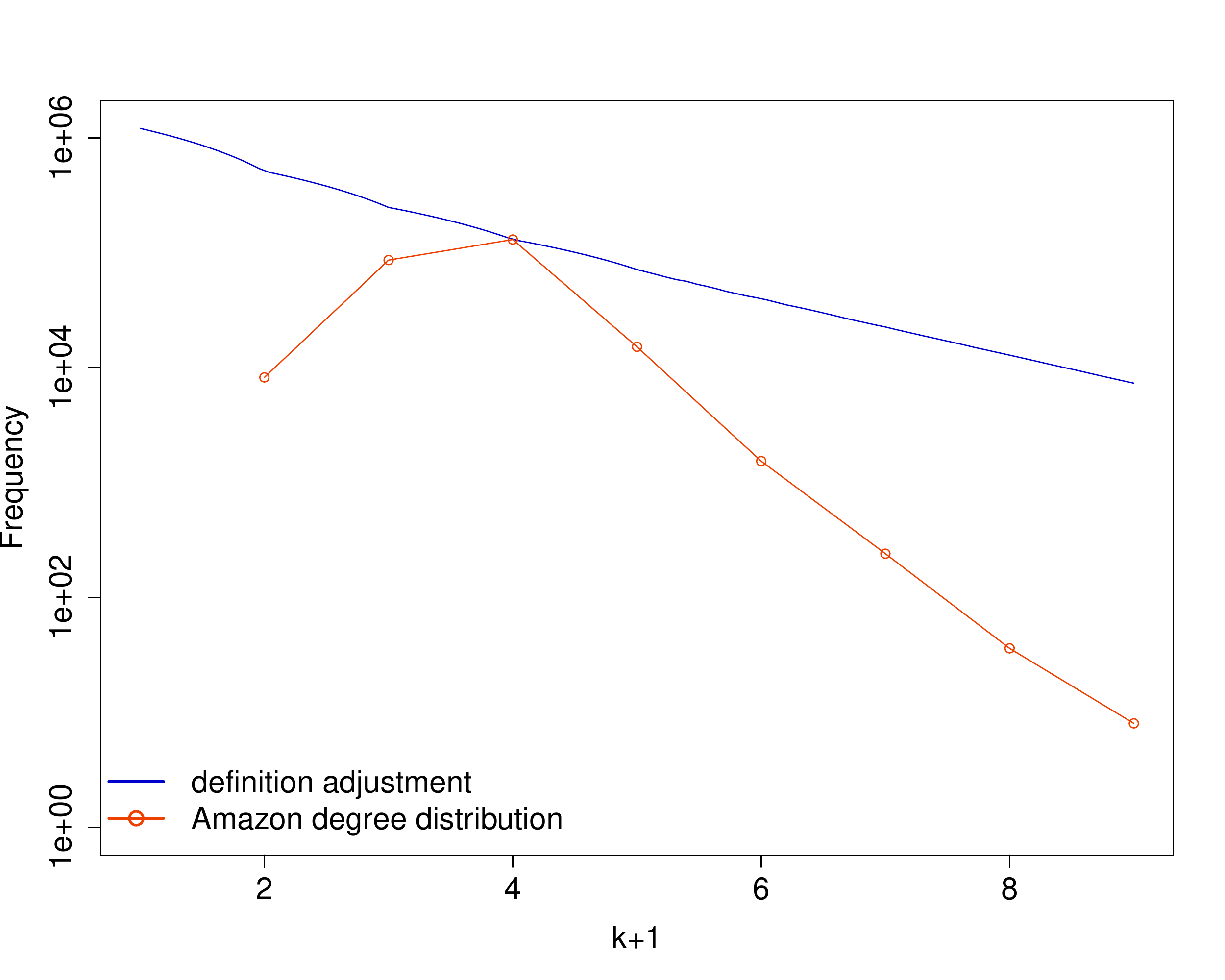}}
\subfigure[\bf{CondMatt (undirected)}]{\includegraphics[scale=0.32]{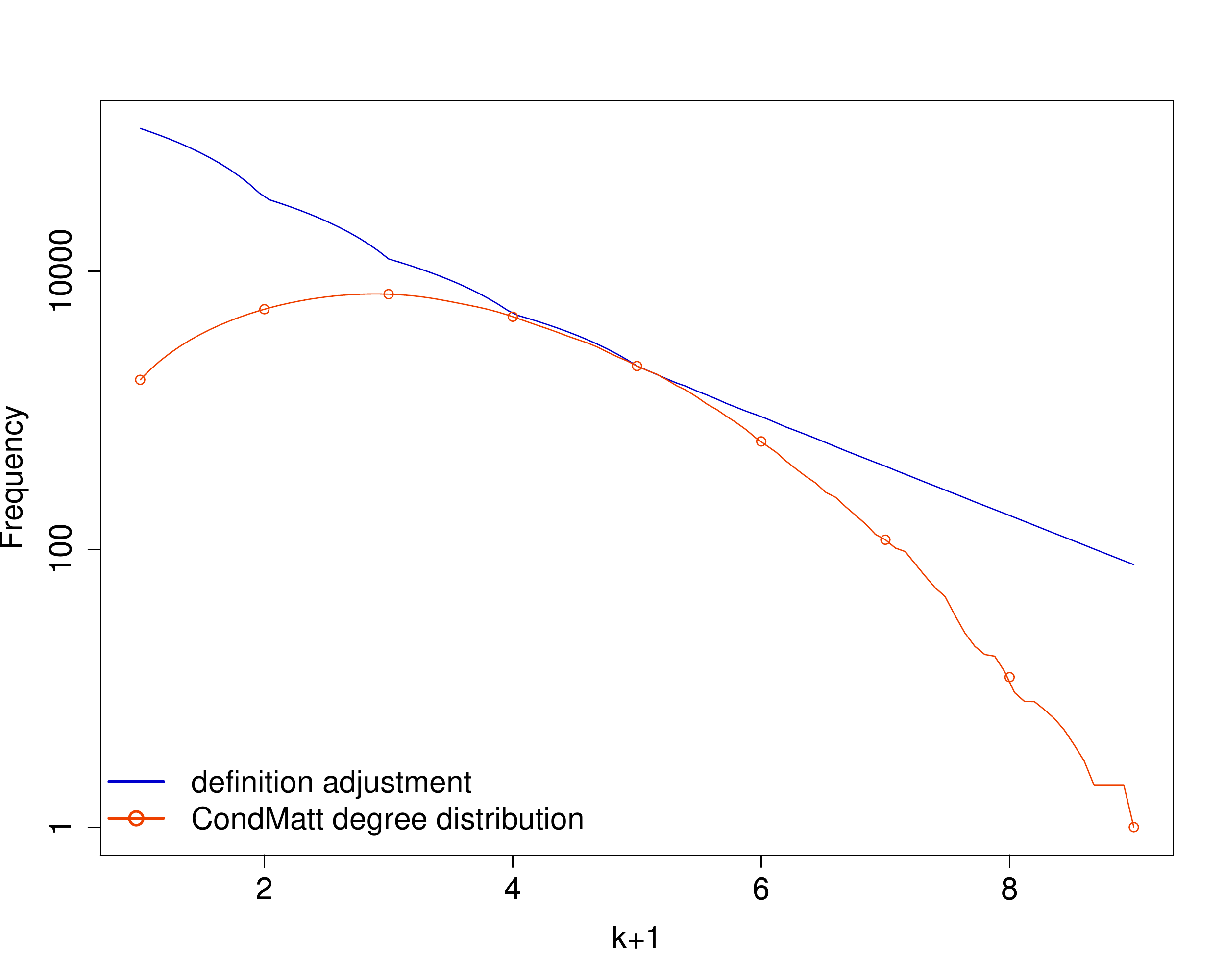}}
\caption{Real-World networks are \dpl{}: definition adjustment for t=0}
\label{fig:dpl_graphs_t0}
\end{figure*}

\begin{figure*}
\centering
\footnotesize
\includegraphics[scale=0.32]{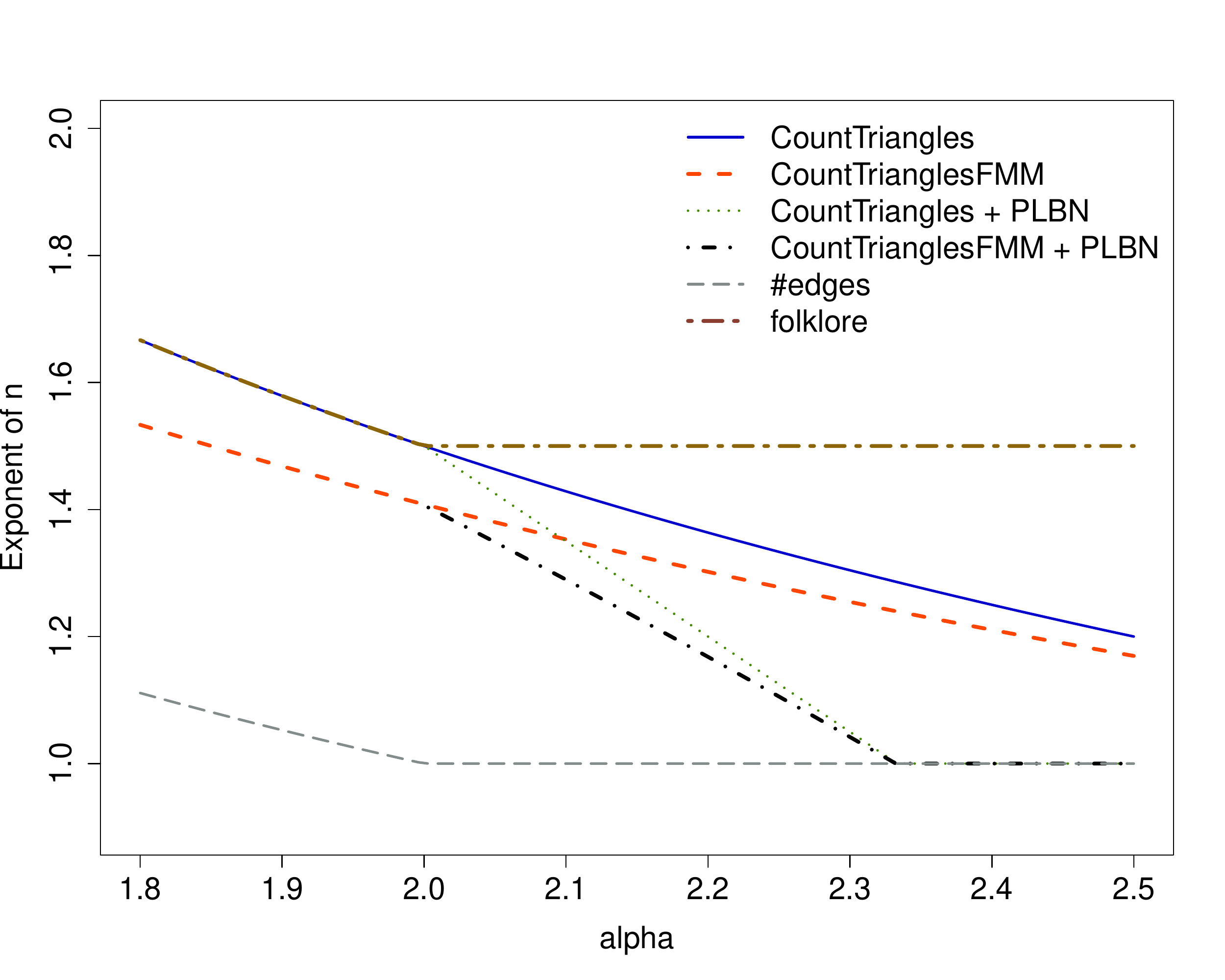}
\caption{The exponent of the running time of our algorithms for counting triangles.
Here PLBN stands for PLB neighborhoods.
\#edges is the number of edges in a graph, and folklore is \#edges multiplied by $3/2$, as the well-known algorithm for counting triangles runs in $O(m^{3/2})$ time.}
\label{fig:times}
\end{figure*}

\begin{figure*}
\centering
\footnotesize
\includegraphics[scale=0.32]{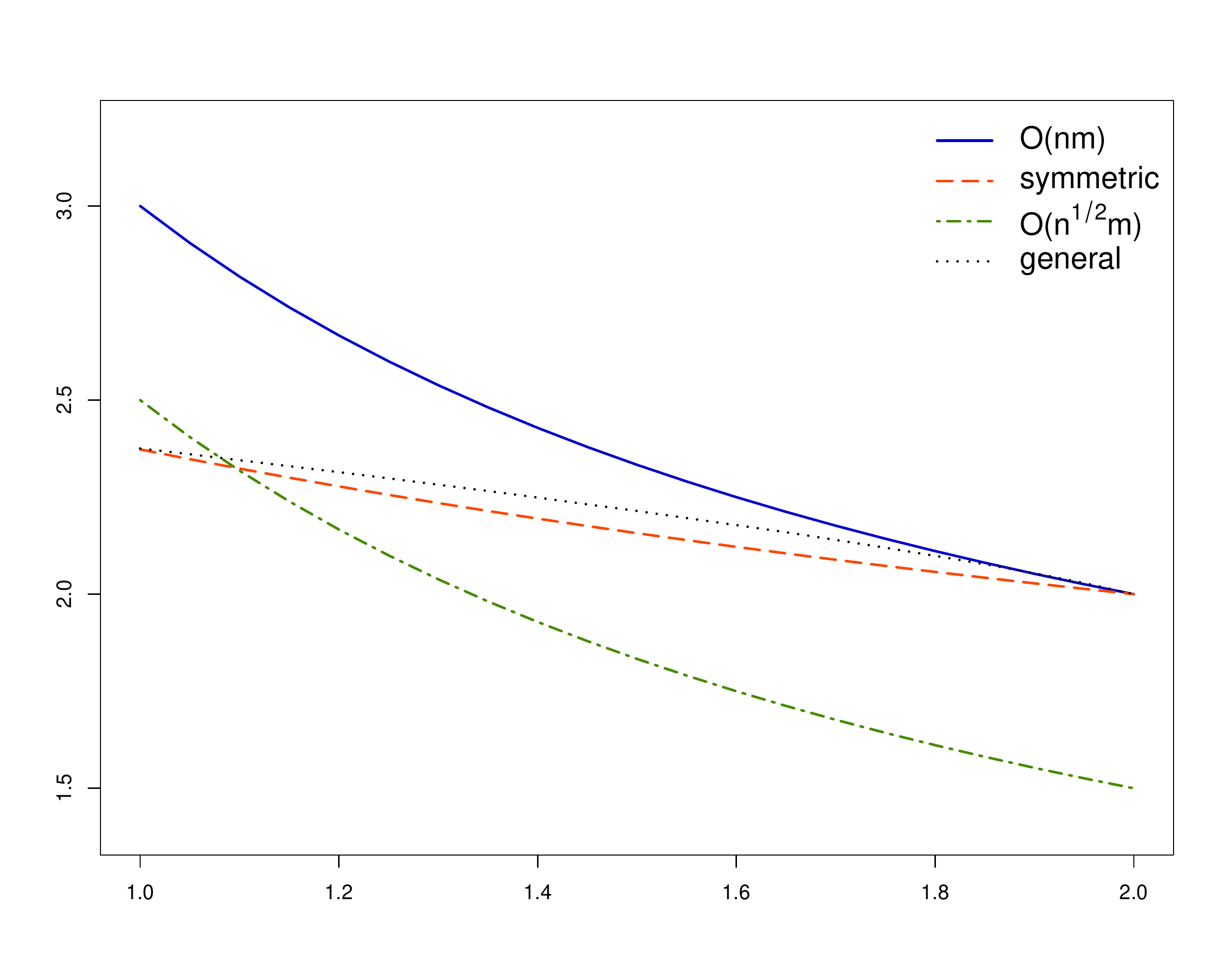}
\caption{The exponent of the running times of our algebraic algorithms for power law graphs and matrices, whose nonzero entries correspond to the edges of a power law graph.
Symmetric shows the complexity of determinant algorithm for symmetric matrices as well as perfect matching algorithm.
General depicts the complexity of algorithms for determinant, PageRank, matrix inverse, linear system solving and 
transitive closure in matrices that do not need to be symmetric.
The complexities are derived using the bound on $\omega(n,n,n_k)$ given in~\cite{6375330}.
Our results are compared to the running times of algorithms that work for arbitrary graphs and matrices.
Note that the bound of $O(\sqrt{n}m)$ is only known for the perfect matching algorithm.}
\label{fig:sank_times}
\end{figure*}

\end{document}